\documentclass[11pt,letterpaper]{article}
\usepackage[margin=1in]{geometry}
\usepackage[T1]{fontenc}
\usepackage{lmodern}
\usepackage{amsmath,amssymb,amsthm,mathtools}
\usepackage{booktabs}
\usepackage{microtype}
\usepackage{enumitem}
\usepackage[dvipsnames]{xcolor}
\usepackage{color-edits}
\addauthor[Mengfan]{mf}{ForestGreen}    %
\usepackage[numbers,square]{natbib}
\usepackage[colorlinks=true,linkcolor=blue!70!black,citecolor=blue!70!black,
  urlcolor=blue!70!black,breaklinks=true]{hyperref}
\usepackage{aliascnt}
\usepackage[capitalise,nameinlink]{cleveref}

\allowdisplaybreaks[1]
\setlist{topsep=3pt,itemsep=2pt}

\newtheorem{theorem}{Theorem}[section]
\newcommand{\sharedtheorem}[2]{%
  \newaliascnt{#1}{theorem}%
  \newtheorem{#1}[#1]{#2}%
  \aliascntresetthe{#1}%
  \crefname{#1}{#2}{#2s}}
\sharedtheorem{lemma}{Lemma}
\sharedtheorem{corollary}{Corollary}
\sharedtheorem{proposition}{Proposition}
\theoremstyle{definition}
\sharedtheorem{definition}{Definition}
\sharedtheorem{example}{Example}
\theoremstyle{remark}
\sharedtheorem{remark}{Remark}
\crefname{equation}{}{}
\Crefname{equation}{}{}

\newcommand{\TPS}{\operatorname{TPS}}
\newcommand{\MMS}{\operatorname{MMS}}
\newcommand{\E}{\mathbb{E}}
\newcommand{\ind}{\mathbf{1}}
\newcommand{\deq}{\triangleq}
\newcommand{\Lottery}{\mathcal{L}}
\newcommand{\AllocationSet}{\mathcal{A}}
\newcommand{\ProbabilityDistributionSet}{\Delta}
\newcommand{\xhdr}[1]{\paragraph{#1}}
\makeatletter
\newcommand{\customlabel}[2]{\begingroup\let\customlabel@href\@currentHref
  \phantomsection\protected@edef\@currentlabel{#2}\ltx@label{{#1}}%
  \global\let\@currentHref\customlabel@href\endgroup}
\makeatother

\hypersetup{pdftitle={Truthful-in-Expectation Mechanism with Constant Maximin-Share Guarantee}}

\title{Truthful-in-Expectation Mechanism with Constant\\ Maximin-Share Guarantee}
\author{%
Mengfan Ma\thanks{Central China Normal University.
  Email: \texttt{mengfanma1@gmail.com}}
\and
Biaoshuai Tao\thanks{Shanghai Jiao Tong University.
  Email: \texttt{bstao@sjtu.edu.cn}}
\and
Fangxiao Wang\thanks{The Hong Kong Polytechnic University.
  Email: \texttt{fangxiao.wang@connect.polyu.hk}}}
\date{}

\begin{document}

\maketitle
\thispagestyle{empty}

\begin{abstract}
We study the truthful and fair allocation of $m$ indivisible goods to $n$
strategic agents with additive valuations. Babaioff, Feige, and Manaker
Morag (FOCS~2026) gave a randomized mechanism that uses only the agents'
rankings of the goods, is truthful-in-expectation (TIE), and guarantees every
agent $1/(H_{n-1}+2)=\Theta(1/\log n)$ of her maximin share (MMS) in every
realized allocation, where $H_{n-1}$ is the $(n-1)$th harmonic number; this
is nearly the best possible with rankings alone. They conjectured that
cardinal information allows TIE mechanisms to achieve a constant ex-post MMS
guarantee. We confirm this conjecture: our TIE mechanism gives every agent
at least $1/6$ of her truncated proportional share (TPS), and hence of her
MMS, in every realized allocation. Moreover, the mechanism is ex-ante
envy-free and can be implemented in polynomial time as an
explicit lottery over at most $m(n-1)+1$ allocations.

\end{abstract}

\newpage
\setcounter{tocdepth}{2}
\tableofcontents
\thispagestyle{empty}
\newpage
\setcounter{page}{1}

\section{Introduction}
\label{sec:introduction}

Fair division is the problem of allocating resources fairly among agents
with heterogeneous preferences. Its central questions concern which
fairness guarantees are achievable and how to compute allocations
satisfying them. Research on these questions has led to algorithms for
practical allocation problems, such as assigning course seats to students
\citep{budish2011combinatorial} and dividing inherited goods among heirs
\citep{goldmanProcaccia2014spliddit}. Fair-division principles also inform
research in machine learning~\citep{chaudhuryEtAl2022federated}, cloud
computing~\citep{wangLiLiang2014cloud}, recommender systems
\citep{patroEtAl2020fairrec}, and participatory budgeting
\citep{fainGoelMunagala2016core}.

A natural fairness notion is \emph{proportionality}: each of the $n$
agents should receive at least a $1/n$ fraction of her value for all
items. With indivisible items, however, a proportional allocation need not
exist. For example, suppose there are two agents and only one good, which
both agents value positively. Whichever agent does not receive the good gets
nothing, whereas proportionality requires that she receive half of her value
for it. This motivates relaxations of proportionality that are
suited to indivisible items. Budish~\citep{budish2011combinatorial}
introduced the \emph{maximin share} (MMS), which has become the most widely
studied such relaxation. To determine her MMS, an agent imagines
partitioning the items into $n$ bundles and receiving the bundle she values
least; she chooses the partition that makes this worst-case value as large
as possible. An allocation giving every agent her full MMS need not exist,
even with additive valuations, for either goods
\citep{procacciaWang2014fair,kurokawaProcacciaWang2018fair} or chores
\citep{azizEtAl2017chores}, which has motivated a long line of work on
approximate guarantees. For additive goods, the guaranteed fraction has
been raised from $2/3$
\citep{procacciaWang2014fair,kurokawaProcacciaWang2018fair,amanatidisEtAl2017approximation}
to $3/4$~\citep{ghodsiEtAl2018improvement}, then slightly beyond $3/4$
\citep{gargTaki2021improved,akramiGarg2024breaking}, and further to $10/13$
\citep{heidariEtAl2026improved}; a recent preprint reports $7/9$
\citep{huangZhou2025fptas}.

Babaioff, Ezra, and Feige~\citep{babaioffEzraFeige2022bobw} introduced the
\emph{truncated proportional share} (TPS), another relaxation of
proportionality for additive valuations. An agent's TPS is her proportional
share after each good's value is capped at the TPS itself. It is at least
her MMS and at most her proportional share, so every $\alpha$-TPS guarantee
implies the corresponding $\alpha$-MMS guarantee. Moreover, while computing
the MMS is NP-hard, the TPS can be computed in polynomial time, so whether
an allocation meets a TPS guarantee can be verified efficiently.

Agents' valuations are often private: an allocation rule sees what they
report, not what they truly value. Since reports affect the outcome, an
agent may benefit from misrepresenting her preferences. Consider a company
assigning scarce GPU servers to its teams. One team has a slower but
workable alternative, yet claims to have none to gain priority over a team
that truly lacks one. The manager cannot verify how much each team values
the servers, so the allocation may appear fair according to the reports
while rewarding exaggeration. Such manipulation undermines fairness itself:
fairness guarantees hold with respect to the reported values, so an agent
who reports honestly may end up worse off than one who exaggerates. A
\emph{truthful} mechanism avoids this by making honest reporting optimal
for every agent, whatever the others report.

Deterministic truthful mechanisms, however, are severely limited. For two
agents with additive valuations, Amanatidis, Birmpas, Christodoulou, and
Markakis~\citep{abcm2017truthful} characterized the deterministic truthful
mechanisms and showed that none of them can guarantee both agents more than
a $1/\lfloor m/2\rfloor$ fraction of their MMS, where $m\ge2$ is the number
of goods; this bound is tight. Thus, even with two agents, the MMS
guarantee of any deterministic truthful mechanism vanishes as the number of
goods grows. These limitations motivate randomized mechanisms.

Among randomized approaches, an important incentive criterion is
\emph{truthfulness-in-expectation} (TIE). A TIE mechanism uses the reports
to select a distribution over allocations; for any reports of the other
agents, reporting her true valuation gives an agent at least as much
expected value as any misreport. A lie may help in a particular outcome but
cannot improve her expected value. TIE appears in early mechanism-design
work by Archer and Tardos~\citep{archerTardos2001truthful} and was applied
to fair division by Mossel and Tamuz~\citep{mosselTamuz2010truthful}. For
indivisible goods, Babaioff, Ezra, and Feige~\citep{babaioffEzraFeige2022bobw}
were the first to study TIE mechanisms, and they raised the question of
which ex-post fairness guarantees such mechanisms can achieve. This leads
to the following question.
\begin{quote}
\emph{How much fairness can a TIE mechanism guarantee in every realized
allocation?}
\end{quote}

For additive goods, Bu and Tao~\citep{buTao2025truthful} gave a TIE
mechanism guaranteeing each agent $1/n$ of her MMS in ex-post. Babaioff,
Feige, and Manaker Morag~\citep{bfmm2026tie} improved this to
$1/(H_{n-1}+2)=\Theta(1/\log n)$ of the TPS using only the agents' rankings
of the goods, where $H_k=\sum_{j=1}^{k}1/j$. This is close to the limit of
ordinal information: Amanatidis, Birmpas, and
Markakis~\citep{amanatidisBirmpasMarkakis2016truthful} proved that no
ordinal algorithm, even without a truthfulness requirement, can guarantee
more than $1/H_n$ of the MMS. Babaioff, Feige, and Manaker
Morag~\citep{bfmm2026tie} also obtained $\Omega(1/\log\log n)$ by using
limited cardinal information, but that mechanism is only \emph{almost}
TIE: reporting truthfully guarantees an agent at least a
$(1-\varepsilon(n))$ fraction of the expected value she could obtain by any
misreport, where $\varepsilon(n)=n^{-\log n}$. A constant guarantee must
therefore use information beyond rankings, and the open question is
whether this can be done while preserving TIE for an arbitrary number of
agents. Babaioff, Feige, and Manaker Morag conjectured that it can. This is
the question we resolve.

\begin{table}[t]
\centering
\small
\begin{tabular}{@{}llll@{}}
\toprule
Result & Incentives & Ex-post fairness & Ex-ante fairness \\
\midrule
Bu and Tao~\citep{buTao2025truthful}
  & TIE & $1/n$-MMS & envy-free \\
Babaioff, Feige, and Manaker Morag~\citep{bfmm2026tie}
  & TIE & $1/(H_{n-1}+2)$-TPS & proportional \\
\addlinespace
This work & TIE & $1/6$-TPS & envy-free \\
\bottomrule
\end{tabular}
\caption{Guarantees of truthful-in-expectation mechanisms for allocating
indivisible goods among $n\ge2$ agents with nonnegative additive
valuations. TPS guarantees imply the corresponding MMS guarantees, since
the TPS is at least the MMS.}
\label{tab:comparison}
\end{table}

\subsection{Our contribution}
\label{subsec:contribution}

We obtain the first constant ex-post MMS guarantee for TIE mechanisms with
an arbitrary number of agents.

\begin{theorem}[Main theorem]
\label{thm:main}
Consider $m$ indivisible goods and $n$ agents with nonnegative additive
valuations. There is a randomized mechanism that is truthful-in-expectation and ex-ante envy-free, and such that every allocation $A$ in the
support of its output satisfies, for every agent $i$,
\[
  v_i(A_i)\;\ge\;\frac{\TPS_i}{6}\;\ge\;\frac{\MMS_i}{6}.
\]
For rational valuations encoded in binary, an explicit lottery over at most
$m(n-1)+1$ complete allocations implementing the mechanism can be computed
in polynomial time.
\end{theorem}

The factor $1/6$ depends on neither the number of agents nor the number of
goods. As \cref{tab:comparison} shows, for general numbers of agents, Bu
and Tao~\citep{buTao2025truthful} guarantee $1/n$ of the MMS, and Babaioff,
Feige, and Manaker Morag~\citep{bfmm2026tie} guarantee $1/(H_{n-1}+2)$ of the
TPS; for two agents, they guarantee $2/3$ of the TPS
\citep[Theorem~3]{bfmm2026tie}. The fractional rule underlying our
mechanism depends only on the agents' orders of the goods, while its
implementation uses the reported values; by the ordinal impossibility of
Amanatidis, Birmpas, and Markakis~\citep{amanatidisBirmpasMarkakis2016truthful},
such use of cardinal information is necessary for any constant guarantee.

\subsection{Technical overview}
\label{subsec:overview}

We do not fix the approximation ratio in advance: the mechanism is analyzed for an $\alpha\,\TPS_i$ guarantee with a parameter $\alpha$, and the analysis yields $\alpha=1/6$, which is tight for this analysis (\cref{subsec:tight}). Part 1 defines a fractional rule $x$ that is truthful-in-expectation. Parts 2 and 3 design a distribution over integral allocations that implements $x$ in expectation. Part 2 decides what each agent receives: high-value goods are reserved by a random matching, and each agent who receives none of them is compensated with a fraction of her low-value goods. Part 3 shows how this is realized: the reservations are implemented by a single total-unimodularity argument (\cref{lem:reservation}), and the compensation is then rounded to integral bundles (\cref{lem:grid}), with the remaining goods assigned so that every marginal equals $x$ (\cref{lem:completion}).

\xhdr{Part 1: the fractional rule.}
Each agent's top set $T_i$ consists of her $n-1$ most valuable goods, with ties broken by index; we say that she \emph{nominates} these goods. Let $f_g$ be the number of agents nominating $g$. At a high level, our rule, the \emph{top-set competitive rule}, raises the probability that agent $i$ receives each good in her top set $T_i$, while keeping the rule truthful-in-expectation. Concretely, the probabilities are:
\begin{equation}\label{eq:intro-x}
 x_{ig}=\underbrace{\frac1n}_{\text{\emph{base probability}}}
       +\underbrace{\frac{\ind[g\in T_i]}n}_{\text{\emph{bonus probability}}}
       -\underbrace{\frac{\sum_{j\in N\setminus\{i\}}\ind[g\in T_j]}{n(n-1)}}_{\text{\emph{penalty probability}}}.
\end{equation}
Each agent receives a base probability $1/n$ of every good, a bonus $1/n$ on every good she nominates, and a penalty $1/(n(n-1))$ for every other agent nominating the good.
For additive valuations, agent $i$'s expected value under any lottery implementing $x$ is $v_i(x_i)$, and by \eqref{eq:intro-x},
\[
  v_i(x_i)=\frac{v_i(M)}{n}+\frac{v_i(T_i)}{n}-\frac{1}{n(n-1)}\sum_{j\in N\setminus\{i\}}v_i(T_j).
\]
Her report determines only her own top set $T_i$, which appears only in the second term; the last term depends only on the other agents' reports. Among all sets of $n-1$ goods, $v_i(T_i)$ is largest for her true top set, so reporting truthfully maximizes her expected value. 
Hence, every lottery implementing $x$ is truthful-in-expectation (\cref{prop:truthful}).
For envy-freeness, compare agent $i$'s probabilities with those of another agent $j$. 
The base terms cancel, and so do the penalties from all agents other than $i$ and $j$, which gives $x_{ig}-x_{jg}=(\ind[g\in T_i]-\ind[g\in T_j])/(n-1)$. 
Therefore,
\[
  v_i(x_i)-v_i(x_j)=\frac{v_i(T_i)-v_i(T_j)}{n-1}\ge0.
\]
Hence, every lottery implementing $x$ is ex-ante envy-free (\cref{prop:truthful}).

\xhdr{Part 2: allocating high goods and compensating dummy agents.}
We implement $x$ by \emph{rounding}: we choose a random integral allocation whose marginal probabilities are exactly $x$.
The standard rounding~\ref{cond:PR} gives each agent, in every prefix of her value order, a number of goods equal to the fractional count rounded down or up; by \cref{lem:prefix}, her realized value is then her expected value minus at most the value of her most valuable good.
Since a single good can be worth as much as $\TPS_i$, this can leave her far below her TPS (\cref{ex:why-high}).
We therefore treat high-value goods separately.
A good is \emph{high} for agent $i$ if she values it at least $\alpha\,\TPS_i$, and the other goods are \emph{low} (\cref{def:high}).
The first step of our rounding allocates high goods integrally by a random matching, whose probabilities are given by the \emph{retained high-good suballocation} $h$ (\cref{def:h}).
$h$ is given by $x$ truncated on the high goods, ensuring that every agent's retained mass is at most $1$.
An agent having mass $1$ under $h$ is allocated a high good with probability $1$, and her TPS requirement is immediately met even without further allocation.
On the other hand, an agent having mass less than $1$ receives no high good with a positive probability $p_i$.
These agents are called \emph{deficient} (\cref{def:deficient}).
A deficient agent not receiving a high good needs more allocations in the later rounding steps for compensation.
We call such a deficient agent a \emph{dummy agent}.

A dummy agent $i$ receives a fraction $1/k_i$ of each good in her compensation set of low goods, where $k_i$, agent $i$'s compensation parameter, is the smallest integer that is at least $2$ and at least $2n$ times the probability $p_i$ that she becomes a dummy agent (see \eqref{eq:parameters}).
This choice lets all dummy agents be compensated at the same time.
Since $k_i\ge 2np_i$, agent $i$'s expected share of each good is $p_i/k_i\le 1/(2n)$; summed over the at most $n$ deficient agents, the expected total share of any good is at most $1/2$.
Since $k_i\ge2$, the share of a single dummy agent is at most $1/2$.
The constraints spreading out dummy agents in Part~3 keep the realized total at most its expectation plus the largest single share, so in every outcome the dummy agents' shares of each good sum to at most $1/2+1/2=1$ (\cref{lem:reservation}).

A dummy agent cannot be compensated with goods that are too popular, since her own probability $x_{ig}$ of such a good is too small to absorb her expected share $p_i/k_i$, nor with goods reserved for others in the same outcome. To avoid these, she gives up a \emph{popular set}: a prefix of the popularity order, namely the most popular goods, where the goods are sorted by how many agents nominate them (\cref{def:prefixes}). The popular set is chosen to minimize the number of low goods she can lose, and this number is bounded (\cref{def:comparison}).

\xhdr{Part 3: reservations via total unimodularity.}
Besides the matching constraints (each agent either receives one high good or becomes a dummy agent, and each good is used at most once), the matching must satisfy two counting constraints:
\begin{itemize}
\item \emph{Popular sets.} For every prefix of the popularity order, the number of goods reserved outside it is capped at the ceiling of its expected value, so every deficient agent loses only a bounded number of low goods to reservations, simultaneously for all agents (constraint~\ref{cons:C}).
\item \emph{Spreading out dummy agents.} Order the deficient agents by decreasing $1/k_i$. For every $\ell$, the number of dummy agents among the first $\ell$ of them is kept within one of its expectation, so the dummy agents' shares of any good stay within its capacity (constraint~\ref{cons:D}).
\end{itemize}
The supports of these constraints form two laminar families, so the constraint matrix is totally unimodular~\citep{edmonds1970submodular,budishCheKojimaMilgrom2013random} (\cref{lem:tu}). Hence, the polytope has integral vertices, and sampling a vertex realizes the retained high goods and the probabilities $p_i$ exactly (\cref{lem:reservation}).

\xhdr{Guarantees.}
In \emph{ex-post}, an agent matched to a high good keeps it and is done. Her low goods are worth at least $a_i\,\TPS_i$ in total (\cref{lem:lowsupply}), where $a_i=n-|H_i|$ is the number of high goods she is short of. Let $d_i$ denote the bound on the number of low goods a dummy agent can lose (\cref{def:comparison}). A dummy agent loses at most $d_i$ low goods in the reservation and less than $(1-1/k_i)\,\alpha\,\TPS_i$ in rounding (\cref{lem:grid}), so she receives at least $\alpha\,\TPS_i$ whenever $\alpha(d_i+2k_i-1)\le a_i$ (see \eqref{eq:denominator}). Total unimodularity gives the cap on $d_i$ (\cref{lem:reservation}), and to show that it is small enough we perform a case analysis on three popular sets, $P_{n-1}$, $P_n$, and the goods nominated by more than half of the agents (\cref{sec:analysis}).

Babaioff, Feige, and Manaker Morag~\citep{bfmm2026tie} and Bu and Tao~\citep{buTao2025truthful} decompose the high-good probabilities into matchings via bipartite edge colorings. We instead take a single integral vertex of one polytope. Besides the matching constraints, it imposes the two counting constraints: the popular-set constraints and the constraints spreading out dummy agents in Part~3. Total unimodularity makes every vertex integral, so both kinds of constraints hold at once. The popular-set constraints bound the low goods a deficient agent can lose to reservations, and the constraints spreading out dummy agents keep each good's compensation within capacity. Together with the choice of popular sets and compensation parameters $k_i$, they yield the constant.

\subsection{Related work}
\label{subsec:related-work}

\xhdr{Truthful mechanisms.}
For divisible resources, Chen, Lai, Parkes, and
Procaccia~\citep{chenEtAl2013cake} give truthful and envy-free cake-cutting
mechanisms under piecewise-uniform valuations. Bei, Huzhang, and
Suksompong~\citep{beiEtAl2020disposal} obtain truthful, envy-free complete
allocations for two agents in both cake cutting and divisible chore
division under the same preference restriction, while
Francis~\citep{francis2022chores} obtains truthfulness and proportionality
for any number of agents with piecewise-uniform chore costs. In contrast,
no deterministic mechanism is both truthful and proportional, even for two
agents with piecewise-constant valuations~\citep{buSongTao2022cake}.
Randomized cake-cutting mechanisms can combine TIE with envy-freeness in
every realization~\citep{chenEtAl2013cake,mosselTamuz2010truthful}.

For indivisible goods, Amanatidis, Birmpas, and
Markakis~\citep{amanatidisBirmpasMarkakis2016truthful} study truthful MMS
approximation, and Amanatidis, Birmpas, Christodoulou, and
Markakis~\citep{abcm2017truthful} characterize deterministic truthful
mechanisms for two additive agents when all goods must be allocated. Babaioff and Manaker
Morag~\citep{babaioffManakerMorag2025characterizations} study
characterizations and fairness for arbitrary numbers of agents under
non-bossiness and neutrality. Restricted valuation domains, including
binary additive and matroid-rank valuations, admit truthful mechanisms with
fairness and efficiency guarantees
\citep{halpernEtAl2020binary,babaioffEzraFeige2021dichotomous}.

For indivisible chores, Aziz, Li, and Wu~\citep{azizLiWu2024chores} give an
ordinal TIE mechanism for general additive costs whose expected maximum MMS
cost ratio is $O(\sqrt{\log n})$; this is an expected approximation
guarantee, not a bound on every realization. Sun and
Chen~\citep{sunChen2025randomized} obtain group-strategyproofness in
expectation together with ex-ante envy-freeness and ex-post EF1 when each
item's cost is either zero or a publicly known item-specific cost. For
common positive bi-valued costs, Li, Tao, Wang, Wu, Yang, and
Zhou~\citep{liEtAl2025truthfulChores} obtain TIE, ex-ante envy-freeness and
Pareto optimality, and ex-post EF1, which also implies an ex-post
$(2-1/n)$-MMS guarantee.

The closest results to ours concern TIE under arbitrary additive goods
valuations. Bu and Tao~\citep{buTao2025truthful} achieve ex-post EF1 for two
agents and, for arbitrary $n$, ex-ante envy-freeness together with ex-post
PROP1 and $1/n$-MMS. Babaioff, Feige, and Manaker
Morag~\citep{bfmm2026tie} obtain an ordinal TIE mechanism with ex-ante
proportionality and ex-post $1/(H_{n-1}+2)$-TPS; they also obtain
$2/3$-TPS under TIE for two agents and study improved guarantees under
almost TIE for general $n$. Our result combines a constant ex-post TPS
guarantee with TIE and ex-ante envy-freeness for arbitrary numbers of
agents.

\xhdr{Share-based fairness.}
An allocation giving every agent her MMS need not exist
\citep{kurokawaProcacciaWang2018fair}; Feige, Sapir, and
Tauber~\citep{feigeSapirTauber2021tight} show that no allocation can
guarantee more than $39/40$ of the MMS for three agents and more than
$1-1/n^4$ for $n\ge4$. On the positive side, the initial $2/3$ MMS
guarantee~\citep{procacciaWang2014fair} was improved through a series of
works
\citep{ghodsiEtAl2018improvement,gargTaki2021improved,akramiGarg2024breaking,heidariEtAl2026improved};
the revised preprint of Huang and Zhou~\citep{huangZhou2025fptas}
establishes the current best general guarantee of $7/9$ and a
$(7/9-\varepsilon)$ FPTAS. With only ordinal information, no algorithm can
guarantee more than a $\Theta(1/\log n)$ fraction of the MMS
\citep{amanatidisBirmpasMarkakis2016truthful}. For chores, the initial
cost factor of $2$~\citep{azizEtAl2017chores} was improved by subsequent
work~\citep{barmanKrishnamurthy2020approximation,huangLu2021chores} to the
current general bound of $13/11$, with a $(13/11+\varepsilon)$
FPTAS~\citep{huangSegalHalevi2023scheduling}.

Babaioff, Ezra, and Feige~\citep{babaioffEzraFeige2022bobw} introduced the
TPS and established the tight TPS approximation ratio $n/(2n-1)$, together
with a polynomial-time algorithm attaining it. Other shares have been
proposed for unequal entitlements, where the weighted MMS admits only a
$1/n$ guarantee~\citep{farhadiEtAl2019asymmetric} but the TPS and the
\emph{anyprice share} (APS) still admit constant ones
\citep{babaioffEzraFeige2024aps,babaioffFeige2025shares}. The APS is at
least the MMS; for additive goods, every agent can be guaranteed a
$2n/(3n-1)$ fraction of her APS with equal entitlements and a $3/5$
fraction with arbitrary entitlements, although, like the MMS, the APS is
NP-hard to compute. Babichenko, Feldman, Holzman, and
Narayan~\citep{babichenkoEtAl2024quantile} propose quantile shares, which
compare a bundle with a uniformly random one and remain feasible for broad
classes of valuations.

\xhdr{Randomized mechanisms for fair division.}
The best-of-both-worlds literature studies which guarantees in every
realized allocation can coexist with fairness in expectation. Aziz, Freeman,
Shah, and Vaish~\citep{azizEtAl2024bobw} obtain ex-ante envy-freeness and
ex-post EF1 for additive goods and, through an adaptation, additive chores.
For goods, Babaioff, Ezra, and Feige~\citep{babaioffEzraFeige2022bobw}
combine ex-ante proportionality with ex-post $1/2$-TPS. These results do
not impose incentive compatibility. Further work considers goods with
unequal entitlements
\citep{azizGangulyMicha2023entitlements,hoeferEtAl2024entitlements}, and
\citep{azizEtAl2026mixed} establishes ex-ante envy-freeness and ex-post
EF1 for additive mixed-sign utilities. The implementation of such
distributions connects to the utility-controlled decomposition of
fractional assignments of Budish, Che, Kojima, and
Milgrom~\citep{budishCheKojimaMilgrom2013random}, which also rests on
bihierarchical (two-laminar-family) constraint structures. For additive
valuations, preserving item-assignment marginals preserves expected
utilities, which allows the lottery to be adjusted while retaining the
fractional rule's incentive and ex-ante fairness properties; this approach
already appears in earlier TIE mechanisms
\citep{buTao2025truthful,bfmm2026tie}.

\subsection{Organization of Paper}

\Cref{sec:prelim} gives the model and the benchmarks. \Cref{sec:design} introduces the fractional rule and the treatment of high goods and compensation (Parts~1 and~2). \Cref{sec:implementation} realizes the integral allocation via total unimodularity (Part~3). \Cref{sec:proof} proves the main theorem, with \cref{sec:marginals} establishing the ex-ante guarantees and \cref{sec:analysis} the ex-post guarantee. \Cref{sec:necessity} explains why the two counting constraints of Part~3, which control dummy agents and popular sets, are needed.

\section{Preliminaries}
\label{sec:prelim}

\leavevmode We study the problem of allocating a set
$M\deq\{1,\ldots,m\}$ of $m$ indivisible goods to a set $N\deq\{1,\ldots,n\}$ of
$n$ strategic agents whose preferences over the goods are private. The
preferences of each agent $i\in N$ are represented by a nonnegative additive
valuation $v_i\colon 2^M\to\mathbb{R}_{\ge0}$, that is,
$v_i(S)=\sum_{g\in S}v_i(g)$ for every good subset $S\subseteq M$, where we write
$v_i(g)\deq v_i(\{g\})$ for simplicity; in particular, $v_i(\emptyset)=0$.
The set of all nonnegative additive valuations on $M$ is denoted by
$\mathcal{V}_{\mathrm{add}}$. A \emph{valuation profile} is a vector
$v\deq(v_1,\ldots,v_n)\in\mathcal{V}_{\mathrm{add}}^n$, containing one
valuation for each agent.

An \emph{allocation}
$A\deq(A_1,\ldots,A_n)$ is a partition of $M$ into $n$ (possibly empty)
bundles, where $A_i$ is the bundle of agent $i$; we denote the set of
all allocations by $\AllocationSet$. A \emph{fractional
allocation} is a matrix $x\in[0,1]^{n\times m}$ with $\sum_{i\in N} x_{ig}=1$ for
every good $g$; it is a \emph{fractional suballocation} if only
$\sum_{i\in N} x_{ig}\le1$ is required.
An allocation is always \emph{complete}, since it partitions $M$; a fractional allocation is \emph{complete} if $\sum_{i\in N} x_{ig}=1$ for every good $g$.
A feasible fractional suballocation with all entries in $\{0,1\}$ is an \emph{integral suballocation}; equivalently, it is a tuple of pairwise disjoint bundles $(A_i)_{i\in N}$ with $\bigcup_{i\in N} A_i\subseteq M$. Lotteries over integral suballocations, their supports, and the fractional suballocations they implement are defined in the same way.\ We write $x_i(S)\deq\sum_{g\in S}x_{ig}$ and
$v_i(x_i)\deq\sum_{g\in M}x_{ig}v_i(g)$. Our main solution concept
for randomized allocation is \emph{lottery}, which is defined as follows.

\begin{definition}[Lottery]\label{def:lottery}
\leavevmode A \emph{lottery} $\Lottery$ is a distribution over all
possible allocations $\AllocationSet$, where
$\Lottery(A)$ is the probability that the lottery selects allocation $A$.
The \emph{support} of a lottery $\Lottery$ is the set of allocations with
positive probability, i.e., $\{A\in\AllocationSet:\Lottery(A)>0\}$. A
lottery $\Lottery$ \emph{implements} a fractional allocation $x$
if, for every agent $i$ and good $g$, agent $i$ receives $g$ with
probability $x_{ig}$, that is,
$\Pr_{A\sim\Lottery}[g\in A_i]=x_{ig}$.
\end{definition}

\leavevmode Each agent reports a valuation in $\mathcal{V}_{\mathrm{add}}$,
and a \emph{randomized mechanism}
$\mathcal{M}\colon\mathcal{V}_{\mathrm{add}}^n\to\ProbabilityDistributionSet(\AllocationSet)$
maps each reported valuation profile $v\in\mathcal{V}_{\mathrm{add}}^n$ to a
lottery $\mathcal{M}(v)$, where $\ProbabilityDistributionSet(\AllocationSet)$
denotes the set of all lotteries. We seek to design truthful-in-expectation
(TIE) mechanisms, in which reporting her true valuation is a best strategy for
every agent, irrespective of the other agents' reports.

\begin{definition}[Truthfulness-in-expectation]\label{def:tie}
\leavevmode A randomized mechanism $\mathcal{M}$ is
\emph{truthful-in-expectation} if, for every agent $i\in N$, every
true valuation $v_i\in\mathcal{V}_{\mathrm{add}}$, every report
$v_i'\in\mathcal{V}_{\mathrm{add}}$, and every profile
$v_{-i}\in\mathcal{V}_{\mathrm{add}}^{n-1}$ of reports by the other agents,
agent $i$'s expected value from reporting her true valuation $v_i$ is at
least that from reporting $v_i'$:
\[
 \E_{A\sim\mathcal{M}(v_i,v_{-i})}\bigl[v_i(A_i)\bigr]
 \;\ge\;
 \E_{A\sim\mathcal{M}(v_i',v_{-i})}\bigl[v_i(A_i)\bigr].
\]
\end{definition}

\leavevmode By \cref{def:lottery}, whenever a lottery $\Lottery$
implements a fractional allocation $x$, we have
$\E_{A\sim\Lottery}[v_i(A_i)]=v_i(x_i)$ for every
agent $i$, since the valuation $v_i$ is additive. Every lottery $\Lottery$
implements exactly one fractional allocation, namely the one with entries
$x_{ig}=\Pr_{A\sim\Lottery}[g\in A_i]$. Hence, a randomized mechanism is TIE if and only if the
fractional allocation it implements, viewed as a function of the reports,
gives every agent her highest true value when she reports truthfully.

\subsection[Fairness]{Fairness}
\label{subsec:fairness}
\leavevmode Besides truthfulness-in-expectation, we also hope a
mechanism can be fair, and we use the maximin share (MMS) and
truncated proportional share (TPS) as benchmarks for approximate fairness
guarantees. A share specifies a target value using only an agent's own
valuation, the goods, and the number of agents; it does not depend on the
other agents' valuations. An allocation can then be required to give each
agent her share, or a specified fraction of it.

\leavevmode The maximin share, introduced by
Budish~\citep{budish2011combinatorial}, asks what the agent could guarantee
by partitioning the goods herself and receiving a least-valued bundle.

\begin{definition}[Maximin share]\label{def:mms}
\leavevmode The \emph{maximin share} (MMS) of agent $i$ with valuation
$v_i\colon 2^M\to\mathbb{R}_{\ge0}$, when allocating the set of goods $M$
among $n$ agents, is
\[
 \MMS_n(v_i,M)\deq\max_{A\in\AllocationSet}\min_{j\in N}v_i(A_j).
\]
\end{definition}

\leavevmode When the context is clear, we write $\MMS_i$ for
$\MMS_n(v_i,M)$. While the MMS is well-defined beyond additive
valuations, the truncated proportional share (TPS),
introduced by Babaioff, Ezra, and Feige~\citep{babaioffEzraFeige2022bobw},
is defined for additive valuations.

\begin{definition}[Truncated proportional share]\label{def:tps}
\leavevmode The \emph{truncated proportional share} (TPS) of agent
$i$ with additive valuation $v_i\colon 2^M\to\mathbb{R}_{\ge0}$, when
allocating the set of goods $M$ among $n$ agents, is
\begin{equation}\label{eq:tps}
 \TPS_n(v_i,M)\deq\max\left\{t\ge0:\frac{1}{n}\cdot\sum_{g\in M}\min\{v_i(g),t\}=t\right\}.
\end{equation}
\end{definition}

\leavevmode Similarly, we write $\TPS_i$ for $\TPS_n(v_i,M)$.

In words, after capping every good's value at $t$, the proportional share
equals $t$ itself, so no single good contributes more than $\TPS_i$. For
$n=2$ and three goods of value $1$, we have $\MMS_i=1<3/2=\TPS_i$; for
$n=2$ and values $10,1,1$, the proportional share is $6$ but $\TPS_i=2$,
since capping removes the effect of the single large good.

\begin{lemma}\label{lem:tps-mms}
For every agent $i$,
\begin{equation}\label{eq:tpsmms}
 \MMS_i\le\TPS_i\le v_i(M)/n .
\end{equation}
\end{lemma}
\begin{proof}
Let $\phi(t)=\sum_{g\in M}\min\{v_i(g),t\}-nt$. It is continuous, $\phi(0)=0$,
and $\phi(t)<0$ for $t>v_i(M)/n$, so the maximum in \eqref{eq:tps} exists
and is at most $v_i(M)/n$. Let $u=\MMS_i$ and fix an allocation
$A\in\AllocationSet$ attaining it. Each of its bundles has capped value
$\sum_{g\in A_j}\min\{v_i(g),u\}\ge\min\{v_i(A_j),u\}=u$, so
$\phi(u)\ge0$, and continuity gives a zero of $\phi$ at or above $u$.
\end{proof}

Unlike the MMS, which is NP-hard to compute, the TPS can be computed in
polynomial time (\cref{app:implementation}). By \cref{lem:tps-mms}, every
guarantee of a fraction of the TPS implies the same fraction of the MMS.

\begin{definition}[Fairness guarantees]\label{def:fairness}
Let $\alpha>0$. A randomized mechanism $\mathcal{M}$ is
\emph{$\alpha$-TPS ex-post} if, for every reported valuation profile $v$, every
allocation $A$ in the support of $\mathcal{M}(v)$ satisfies
$v_i(A_i)\ge\alpha\,\TPS_i$ for every agent $i$; \emph{$\alpha$-MMS ex
post} is defined analogously. It is \emph{ex-ante envy-free} if, for every
valuation profile $v$ and all agents $i,j$,
$\E_{A\sim\mathcal{M}(v)}[v_i(A_i)]\ge\E_{A\sim\mathcal{M}(v)}[v_i(A_j)]$.
\end{definition}

All guarantees in \cref{def:fairness} are required at every profile, and
hence hold with respect to the true valuations whenever agents report
truthfully. 
Agent $i$'s \emph{proportional share} is $v_i(M)/n$, and ex-ante proportionality
requires $\E[v_i(A_i)]\ge v_i(M)/n$ for every agent $i$. 
Summing the envy-freeness inequalities over $j$ shows that
ex-ante envy-freeness implies ex-ante proportionality.

We will assume $m\geq n\geq2$ in the remaining part of this paper.
The remaining cases are trivial. Without goods, the empty allocation
suffices. If $n=1$, the sole agent receives all goods. If $m<n$, every
TPS is zero, since $\sum_{g\in M}\min\{v_i(g),t\}\le mt<nt$ for $t>0$, and giving
all goods to an agent chosen uniformly at random, independently of the
reports, satisfies \cref{thm:main}. 

\subsection{Polytopes with Integral Vertices}
\label{subsec:polytopes}
To implement a fractional allocation by a lottery, we use a classical
integrality tool. A matrix is
\emph{totally unimodular} if every square submatrix has determinant $-1$,
$0$, or $1$. Given a polytope $\{y\in\mathbb{R}_{\ge0}^{n\times m}:Cy\le c\}$, if the constraint matrix
$C$ is totally unimodular and $c$ is integral, then every vertex of the
polytope is integral~\citep{hoffmanKruskal1956}.\footnote{For a simple proof of this, a vertex of this polytope is given by a $k\times k$ square submatrix $C'$ of $C$ (with $k\leq nm$ for $nm$ variables) and a subvector $c'$ of $c$ with $k$ entries such that $C'y'=c'$, where $y'$ corresponds to $k$ of the $nm$ variables for which we have not yet assign the value $0$. In addition, we have $\det(C')\neq0$ for nondegeneracy. Let $y_i'$ be the $i$-th coordinate of $y'$; by Cramer's rule, we have $y_i'=\frac{\det(C'_i\mid c')}{\det(C')}$ where $C'_i\mid c'$ is the matrix by replacing the $i$-th column of $C'$ by $c'$. The denominator is $\pm1$ by total unimodularity, and the numerator is an integer since $C'_i\mid c'$ is an integer matrix for integral $c$. Thus, each coordinate $y_i'$ is an integer.}  Hence, if the polytope is bounded, every point $y$ in it is a weighted
average of finitely many of its vertices, and sampling each of these vertices
with probability equal to its weight yields a lottery over integral points of
the polytope whose expectation is exactly $y$. This is how we turn a fractional allocation with
constraints into a lottery: the fractional allocation is the point $y$, and
the constraints describe the properties the fractional allocation must have.
If we can show that the constraint matrix is totally unimodular, then the
vertices of the polytope are integral points, namely allocations that satisfy
all the constraints, and a lottery over them implements the fractional
allocation. When the polytope has polynomially many variables and
constraints, such a lottery can be computed in polynomial time
(\cref{app:implementation}).

Total unimodularity is easy to establish when the constraints come from at
most two laminar families. A family of sets is \emph{laminar} if any two of its
members are disjoint or one contains the other. For a family $\mathcal{F}$
of subsets of a ground set $E$, its \emph{incidence matrix} has one row per
member $F\in\mathcal{F}$ and one column per element $e\in E$, with entry $1$
if $e\in F$ and $0$ otherwise. If $\mathcal{F}$ is the union of at most two
laminar families, its incidence matrix is totally unimodular
\citep{edmonds1970submodular}.

Let the ground set be $E=\{(i,g):i\in N,\ g\in M\}$, and let
$x\in\mathbb{R}_{\ge0}^{n\times m}$ be a fractional allocation. The constraint $\sum_{i\in N}x_{ig}=1$ of
good $g$ involves the coordinates in the set $\{(i,g):i\in N\}$. These sets are pairwise
disjoint, so the family $\mathcal{F}_1=\{\{(i,g):i\in N\}:g\in M\}$ is
laminar.
Suppose that $x$ also satisfies several constraints on each agent $i$, one
for each set $S_i$ in a laminar family $\mathcal{S}_i$ of sets of goods, of
the form $\ell\le\sum_{g\in S_i}x_{ig}\le u$. Such a constraint involves the coordinates in the set
$\{(i,g):g\in S_i\}$. The family
$\mathcal{F}_2=\{\{(i,g):g\in S_i\}:i\in N,\ S_i\in\mathcal{S}_i\}$ is
laminar since each $\mathcal{S}_i$ is laminar. Then the
incidence matrix of $\mathcal{F}_1\cup\mathcal{F}_2$ is totally unimodular,
so we can implement $x$ by a lottery over allocations that satisfy all the
constraints.
The following lemma makes this explicit. It combines the two facts
from~\citep{hoffmanKruskal1956} and~\citep{edmonds1970submodular}, and is
known in the random-assignment literature as the bihierarchy condition
\citep{budishCheKojimaMilgrom2013random}. The proof is given in
\cref{app:tu}.

\begin{lemma}\label{lem:tu}
Let $\mathcal{F}_1$ and $\mathcal{F}_2$ be laminar families of subsets of
a finite set $E$. For every $F\in\mathcal{F}_1\cup\mathcal{F}_2$, fix
integers $\ell_F\le u_F$, and consider the polytope
\[
 P=\Bigl\{y\in[0,1]^E:\ \ell_F\le\sum_{e\in F}y_e\le u_F
 \text{ for all } F\in\mathcal{F}_1\cup\mathcal{F}_2\Bigr\}.
\]
Then every vertex of $P$ is integral. Consequently, every point $y\in P$
is a convex combination of finitely many $0$--$1$ points of $P$.
\end{lemma}

\subsection{The rounded-prefix lemma}
\label{subsec:tools}

We conclude this section by introducing another technical lemma, which will be used multiple times.
Consider a Knapsack scenario where we have $m$ items with values $\{v_j\}$ and weights $\{w_j\}$ with $0\leq w_j\leq 1$, and assume items are sorted by decreasing values.
Typically, the items will be the goods sorted by an agent $i$'s values, and the weight of good $g$ will be its marginal $x_{ig}$ for a given fractional allocation $x$.
Suppose we have a subset of items $S$ such that, for every prefix $\{1,\ldots,t\}$ of the ordered $m$ items, the number of items selected by $S$, namely, $|S\cap\{1,\ldots,t\}|$, is close to the total weight of this prefix $\sum_{j=1}^tw_j$.
The following lemma says that the total value of $S$, namely, $\sum_{j\in S}v_j$ is close to the weighted value sum of all items $\sum_{j=1}^sw_jv_j$, up to an additive error of at most $v_1$.
If we only have the upper bound (lower bound, resp.) for $|S\cap\{1,\ldots,t\}|$, then the corresponding upper bound (lower bound, resp.) continues to hold.

\begin{lemma}\label{lem:prefix}
Let $v_1\ge\cdots\ge v_m\ge0$ and $0\le w_j\le1$ for each $j\in\{1,\ldots,m\}$.
Consider a subset $S\subseteq\{1,\ldots,m\}$.
\begin{enumerate}
    \item If $|S\cap\{1,\ldots,t\}|\le\Bigl\lceil\sum_{j=1}^t w_j\Bigr\rceil$ for every $t=1,\ldots,m$, then $\sum_{j\in S}v_j\le\sum_{j=1}^m w_jv_j+v_1$.
    \item If $|S\cap\{1,\ldots,t\}|\ge\Bigl\lfloor\sum_{j=1}^t w_j\Bigr\rfloor$ for every $t=1,\ldots,m$, then $\sum_{j\in S}v_j\ge\sum_{j=1}^m w_jv_j-v_1$.
\end{enumerate}
If all weights are multiples of $1/k$ for a positive integer $k$, the error term $v_1$ in 1 and 2 above can be replaced to $(1-\frac1k)v_1$.
\end{lemma}
\begin{proof}
Let $W_t=\sum_{j=1}^t w_j$ be the total weight of the first $t$ items
and $C_t=|S\cap\{1,\ldots,t\}|$ be the number of selected items among them.
Set the auxiliary value $v_{m+1}=0$.
Then we have
\begin{align}
 \sum_{j\in S}v_j-\sum_{j=1}^m w_jv_j
 &\overset{(a)}{=}
 \sum_{j\in S}\sum_{t=j}^m(v_t-v_{t+1})
 -\sum_{j=1}^m\sum_{t=j}^m w_j\cdot(v_t-v_{t+1}) \notag\\
 &\overset{(b)}{=}
 \sum_{t=1}^m\sum_{\substack{j\in S\\j\le t}}(v_t-v_{t+1})
 -\sum_{t=1}^m\sum_{j=1}^t w_j\cdot(v_t-v_{t+1}) \notag\\
 &\overset{(c)}{=}
 \sum_{t=1}^m(v_t-v_{t+1})\cdot(C_t-W_t),
 \label{eq:summation}
\end{align}
where (a) follows from $v_j=\sum_{t=j}^m(v_t-v_{t+1})$,
(b) exchanges the order of summation, and (c) follows from the definitions
of the selected count $C_t$ and the total weight $W_t$.
For every $t$, the preconditions for 1 and 2 give $C_t-W_t\leq 1$ and $C_t-W_t\geq -1$ respectively.
Substituting these into \eqref{eq:summation} proves both 1 and 2.

Now assume that $k$ is a positive integer and every weight $w_j$ is a multiple
of $1/k$. Then every prefix weight $W_t$ is also a multiple of $1/k$.
For 1, we have $C_t-W_t\le\lceil W_t\rceil-W_t\le1-\frac1k$ for every $t$.
Substituting this into \eqref{eq:summation} gives
$\sum_{j\in S}v_j\le\sum_{j=1}^m w_j\cdot v_j+\left(1-\frac1k\right)\cdot v_1$.
Similarly, we can prove 2.
\end{proof}

\section{The fractional rule, deficient agents, and compensation}
\label{sec:design}

This section sets up the mechanism. 
We first define the fractional rule and show that every mechanism implementing it is truthful-in-expectation and ex-ante envy-free (\cref{sec:rule}).
We then identify the agents for whom an implementation must be designed with care, \emph{the deficient agents}, together with the parameters that describe them (\cref{sec:high}), and finally describe how deficient agents are handled (\cref{sec:compensation}).
The implementation itself is described in \cref{sec:implementation}.
Throughout \cref{sec:design,sec:implementation}, $\alpha\in(0,1)$ denotes the target fraction of the TPS; \cref{sec:analysis} shows that $\alpha=1/6$ is achieved.

\subsection[The top-set competitive rule]{The top-set competitive rule}
\label{sec:rule}

For each agent $i$, let $T_i$ be the set of the $n-1$ goods she values most, with ties broken by index; we call $T_i$ her \emph{top set} and say that $i$ \emph{nominates} the goods in $T_i$. 
Let $f_g=|\{i:g\in T_{i}\}|$ be the number of agents nominating good $g$. 
The \emph{top-set competitive rule} is
\begin{equation}\label{eq:x}
 x_{ig}=\underbrace{\frac1n}_{\text{\emph{base probability}}}
       +\underbrace{\frac{\ind[g\in T_i]}n}_{\text{\emph{bonus probability}}}
       -\underbrace{\frac{\sum_{j\in N\setminus\{i\}}\ind[g\in T_j]}{n(n-1)}}_{\text{\emph{penalty probability}}}.
\end{equation}
Each agent receives a base probability $1/n$ of every good, a bonus $1/n$ on every good she nominates, and a penalty $1/(n(n-1))$ for every other agent nominating the good. 
This rule can be viewed as a special case in the framework of truthful and envy-free mechanisms of Freeman, Witkowski, Vaughan, and Pennock~\citep{freemanWitkowskiVaughanPennock2024equivalence} for divisible goods, with the indicator of the top set in place of their proper scoring rule. 
See Appendix~\ref{sec:properscoringrule} for more details.
The rule in \eqref{eq:x} can be equivalently expressed by $f_g$, the number of agents nominating $g$, as the following lemma shows.
\begin{lemma}\label{lem:marginals}
For every agent $i$ and good $g$, the fractional allocation rule defined in \eqref{eq:x} is equivalent to
\begin{equation}\label{eq:xcases}
 x_{ig}=
 \begin{cases}
 \dfrac{2n-1-f_g}{n(n-1)}\in\Bigl[\dfrac1n,\dfrac2n\Bigr], & g\in T_i,\\[8pt]
 \dfrac{n-1-f_g}{n(n-1)}\in\Bigl[0,\dfrac1n\Bigr], & g\notin T_i.
 \end{cases}
\end{equation}
Moreover, $\sum_{i\in N} x_{ig}=1$ for every good $g$, so $x$ is a complete
fractional allocation, and $\sum_{g\in M} f_g=n(n-1)$.
\end{lemma}
\begin{proof}
If $g\in T_i$, the
agents other than $i$ nominate $g$ exactly $f_g-1$ times, so \eqref{eq:x} gives
\[
 x_{ig}=\frac2n-\frac{f_g-1}{n(n-1)}=\frac{2n-1-f_g}{n(n-1)},
\]
which lies in $[1/n,2/n]$ since $1\le f_g\le n$. If $g\notin T_i$, the
other agents nominate $g$ exactly $f_g$ times, so
\[
 x_{ig}=\frac1n-\frac{f_g}{n(n-1)}=\frac{n-1-f_g}{n(n-1)},
\]
which lies in $[0,1/n]$ since $0\le f_g\le n-1$.

Next, fix a good $g$ and sum \eqref{eq:x} over all agents. The base terms
contribute $n\cdot1/n=1$ and the bonuses contribute $f_g/n$. The
penalty terms contribute
\[
 \sum_{i\in N}\frac{\sum_{j\in N\setminus\{i\}}\ind[g\in T_j]}{n(n-1)}=\sum_{i\in N}\frac{f_g-\ind[g\in T_i]}{n(n-1)}=\frac{nf_g-f_g}{n(n-1)}=\frac{f_g}n.
\]
Hence, $\sum_{i\in N} x_{ig}=1$,
and together with $x_{ig}\ge0$ from the first part, $x$ is a complete
fractional allocation.

Finally,
\[
 \sum_{g\in M} f_g\overset{(a)}{=}\sum_{g\in M}\sum_{i\in N}\ind[g\in T_i]\overset{(b)}{=}\sum_{i\in N}\sum_{g\in M}\ind[g\in T_i]\overset{(c)}{=}n(n-1),
\]
where (a) follows from the definition of $f_g$, (b) exchanges the order of summation, and (c) holds since every top set has $n-1$ goods.
\end{proof}

The rule alone determines incentives and ex-ante fairness: every lottery that implements $x$ inherits both properties.

\begin{proposition}\label{prop:truthful}
Any mechanism that implements the top-set competitive rule in \eqref{eq:x} is truthful-in-expectation and ex-ante envy-free.
\end{proposition}
\begin{proof}
\leavevmode Let $\mathcal{M}$ be a mechanism that implements the fractional rule in \eqref{eq:x}.
Fix an agent $i$ with true valuation $v_i$ and the reports $v_{-i}$ of the
other agents. If agent $i$ reports $v_i'$, let $T_i'$ be the top set
computed from $v_i'$; the top sets of the other agents do not change. By
\eqref{eq:x}, agent $i$'s expected value from reporting $v_i'$ is
\[
 \E_{A\sim\mathcal{M}(v_i',v_{-i})}\bigl[v_i(A_i)\bigr]=\sum_{g\in M} x_{ig}\cdot v_i(g)
 =\frac{v_i(M)}n+\frac{v_i(T_i')}n-\frac{\sum_{j\in N\setminus\{i\}}v_i(T_j)}{n(n-1)} .
\]
Only the middle term depends on her report, and it is maximized by the
truthful report, since $T_i$ consists of $n-1$ goods of maximum total
value under $v_i$. Hence, the mechanism is TIE.

At any profile $v$ and for any agents $i,j$, the base terms and the penalties from agents other than $i,j$ cancel, leaving
$x_{ig}-x_{jg}=(\ind[g\in T_i]-\ind[g\in T_j])/(n-1)$. Consequently,
\[
 \E_{A\sim\mathcal{M}(v)}\bigl[v_i(A_i)\bigr]-\E_{A\sim\mathcal{M}(v)}\bigl[v_i(A_j)\bigr]
 =\frac{v_i(T_i)-v_i(T_j)}{n-1}\ge0,
\]
where the last inequality holds since $T_i$ maximizes agent $i$'s value among sets of size $n-1$.
This shows ex-ante envy-freeness.
\end{proof}

\paragraph{Remarks on the equal division rule.}
A simple and natural fractional allocation rule that satisfies truthfulness-in-expectation and ex-ante envy-freeness is the \emph{equal division rule}, where we have $x_{ig}=\frac1n$ for every pair $(i,g)$.
However, as presented in Babaioff, Feige, and Manaker Morag~\citep{bfmm2026tie}, the following example shows that this rule can guarantee at most $1/n$-MMS in ex-post.

\begin{example}
    There are $n$ agents and $2n-1$ items. Agent $1$ values the first $n-1$ items at $1$ and the last $n$ items at $1/n$. Each of the remaining agents values the first $n-1$ items at $0$ and the last $n$ items at $1$. It is easy to see that $\MMS_i=1$ for every agent. By the equal division rule, the total fraction of the first $n-1$ items received by agent $1$ is $(n-1)/n<1$, which means the expected number of items (from the first $n-1$ items) received by agent $1$ in ex-post is less than $1$.
    Thus, with a positive probability, agent $1$ receives none of the first $n-1$ items.
    On the other hand, each of the remaining $n-1$ agents must receive at least one of the last $n$ items to guarantee a positive MMS share.
    Therefore, in the case agent $1$ receives no item from the first $n-1$, agent $1$ can only receive at most one item from the last $n$ items, which has value only $1/n=(1/n)\MMS_1$.
\end{example}

In the above example, the first $n-1$ items are valuable only for agent $1$, and it is large enough such that one single item guarantees the MMS fairness. On the other hand, the last $n$ items are valuable to the remaining $n-1$ agents, and not-so-valuable to agent $1$.
It is then natural to let agent $1$ receive more in the first $n-1$ items (where no one else is competing), so that agent $1$ is not competing with the remaining agents on the last $n$ items (where the competition is already intense).
This is exactly what our top-set competitive rule does.

\vspace{0.5cm}

The rest of the paper constructs, for every report profile, a lottery that implements $x$ and gives every agent $i$ value at least $\alpha\TPS_i$ in every allocation in the lottery's support for $\alpha=1/6$.
Table~\ref{tab:notations} lists the notations, some of which will be introduced later.

\begin{table}[t]
    \centering
    \begin{tabular}{l p{0.62\linewidth}}
    \hline
       Notation  & Meaning \\
    \hline
        $\alpha$ & the target approximation ratio of TPS; $\alpha=1/6$ in this paper\\
        $\TPS_i$ & agent $i$'s TPS threshold\\
       $x=\{x_{ig}\}_{i\in N;g\in M}$  & the fractional allocation used by our mechanism; see \eqref{eq:x}\\
       $T_i$ & top set: agent $i$'s first $n-1$ largest goods, or equivalently, goods nominated by $i$\\
       $f_g$ & number of nominations of $g$: $f_g=|\{i:g\in T_i\}|$\\
       $P_t$ & first $t$ goods with highest numbers of nominations\\
       $H_i$ & high goods of agent $i$: $H_i=\{g:v_i(g)\geq\alpha\cdot\TPS_i\}$\\
       $L_i$ & low goods of agent $i$: $L_i=M\setminus H_i$\\
       $D$ & set of deficient agents: $D=\{i:\sum_{g\in H_i}x_{ig}<1\}$\\
       $p=\{p_i\}_{i\in N}$ & $i$'s probability of not receiving high good: $p_i=\max\{0,1-\sum_{g\in H_i}x_{ig}\}$\\
       $q_i,a_i,k_i$ & deficient agent only; $q_i=np_i$, $a_i=n-|H_i|$, and $k_i=\max\{2,\lceil 2q_i\rceil\}$\\
       $h=\{h_{ig}\}_{i\in N;g\in M}$ & retained high-good suballocation: the fractional suballocation used in the first step\\
       $\rho(Q)$ & given $Q\subseteq M$, the allocated mass outside $Q$ by $h$: $\rho(Q)=\sum_{i\in N} h_i(M\setminus Q)$\\
       $Q_i,d_i$ & agent $i$'s comparison set, and its cost (Definition~\ref{def:comparison})\\
       $R$ & set of allocated/reserved goods after the first step (implementation of $h$)\\
       $b=\{b_{ig}\}_{i\in N; g\in M}$ & built on the outcome $R$ after the first step, the fractional suballocation in the second step\\
    \hline
    \end{tabular}
    \caption{Table of Notations}
    \label{tab:notations}
\end{table}

\subsection{High goods and deficient agents}
\label{sec:high}

This subsection designs the first step of the mechanism, which handles high goods and deficient agents, and fixes the parameters of the second step. The guarantees of both steps are proved in \cref{sec:analysis}.

A common way to implement a fractional allocation $x$ is to make every realized outcome $(A_1,\ldots,A_n)$ close to $x$ in the way of \cref{lem:prefix}, where the marginal $x_{ig}$ is viewed as the weight.
Specifically, consider each agent $i$, and let $\pi_i(t)$ be the $t$-th largest item according to $v_i$ (with ties broken by index).
We ensure the following \emph{prefix-rounding} property:
\begin{itemize}[topsep=\abovedisplayskip]
    \item[(PR)]\customlabel{cond:PR}{(PR)} for each agent $i$ and each $t\in\{1,\ldots,m\}$, the number of goods selected from $i$'s first $t$ largest goods, $|A_i\cap\{\pi_i(1),\ldots,\pi_i(t)\}|$, is between  $\lfloor\sum_{s=1}^tx_{i,\pi_i(s)}\rfloor$ and $\lceil\sum_{s=1}^tx_{i,\pi_i(s)}\rceil$.
\end{itemize}
This can be simultaneously guaranteed for all the $n$ agents (with different item orders for different agents).
To see this, consider fractional allocations $y$ in the following polytope:
\begin{align*}
    \forall i\in N,t\in M:&\quad \left\lfloor\sum_{s=1}^tx_{i,\pi_i(s)}\right\rfloor\leq\sum_{s=1}^ty_{i,\pi_i(s)}\leq \left\lceil\sum_{s=1}^tx_{i,\pi_i(s)}\right\rceil\\
    \forall g\in M:&\quad \sum_{i\in N}y_{i,g}=1
\end{align*}
The top-set competitive fractional allocation $x$ is clearly in this polytope (replacing $y$ by $x$ satisfies all constraints).
In addition, the coefficient matrix of the constraints is totally unimodular: each of the two types of constraints form a laminar family (see \cref{subsec:polytopes}).
Therefore, each vertex of the polytope is an integral allocation $(A_1,\ldots,A_n)$ satisfying \ref{cond:PR}, and by \cref{lem:tu}, $x$ can be implemented by a lottery over such allocations, so that every allocation in the lottery's support satisfies \ref{cond:PR}.
We remark that this idea is widely used in the past literature~\citep{azizEtAl2024bobw,babaioffEzraFeige2022bobw,buTao2025truthful,bfmm2026tie,azizEtAl2026mixed}.

For our purpose, suppose we just implement the top-set competitive rule $x$ by such a lottery. Lemma~\ref{lem:prefix} shows that every allocation in the lottery's support gives agent $i$ value at least $\sum_{g\in M}x_{ig}v_i(g)-v_i(\pi_i(1))$.
The first term $\sum_{g\in M}x_{ig}v_i(g)$ is agent $i$'s value on $x$, and $\sum_{g\in M}x_{ig}v_i(g)\geq\frac1nv_i(M)$ by proportionality, which is implied by envy-freeness of $x$ (\cref{prop:truthful}).
Since $\TPS_i$ is a lower bound of $\frac1nv_i(M)$, the lower bound on agent $i$'s ex-post value is $\TPS_i-v_i(\pi_i(1))$.
This bound is only useful if the value of the largest item of agent $i$, i.e., $v_i(\pi_i(1))$, is substantially smaller than $\TPS_i$.
However, this may not be the case, as $v_i(\pi_i(1))$ can be exactly $\TPS_i$.
\Cref{ex:why-high} below shows that, as a result, applying \ref{cond:PR} directly to $x$ does not guarantee any constant fraction of the TPS or the MMS.

\begin{example}\label{ex:why-high}
    Consider $n$ agents with identical valuations on $m$ items for $m\gg n$, where the goods are indexed so that $v_i(g)$ is nonincreasing in $g$. All agents therefore order the goods in the same way, namely by index.
    The first $k$ items (with $k<n$) are large where everyone has value $1$, and each of the remaining small items is worth $\frac{n-k}{m-n+1}$ for every agent.
    The MMS threshold is about $1$ for every agent, with the MMS partition such that the first $k$ bundles are the singleton bundles of large items and the last $n-k$ bundles contain the remaining items split evenly.
    For $x$ given by our rule \eqref{eq:x}, we have $x_{ig}=1/n$ for every agent $i$ and every item $g$ and thus, $\sum_{s=1}^{t}x_{i,s}=\frac{t}{n}$.
    The condition \ref{cond:PR} then implies that every agent receives exactly one item from each of the item blocks $\{1,\ldots,n\}$, $\{n+1,\ldots,2n\}$, $\ldots$.
    For the first $n$ items with $k$ large ones, exactly $n-k$ agents do not receive large items.
    Since each of them receives one item from each subsequent block, any agent who does not receive a large item has total value only approximately $\frac{n-k}{n}$, which gives a poor MMS guarantee for large $k$.
\end{example}

For a correct implementation of the example above, for those $n-k$ agents who do not receive any large item, they need to be compensated with small items.
For example, we can select uniformly at random $k$ agents to receive the $k$ large items, and allocate the set $M'$ of small items to the remaining $n-k$ agents by the equal division rule; that is, each small item goes to each of these agents with probability $\frac{1}{n-k}$. Then implement this fractional allocation of $M'$ over the remaining $n-k$ agents by a lottery satisfying \ref{cond:PR}.
Each of these $n-k$ agents receive one item from each block if $M'$ is partitioned into blocks of $n-k$ consecutive items.
Applying Lemma~\ref{lem:prefix} on these $n-k$ agents and item set $M'$, each agent's value is lower bounded by $\frac1{n-k}\cdot v_i(M')-\frac{n-k}{m-n+1}$, which is close to $1$ and has a good MMS guarantee for those who do not receive large items.
The lesson we learned from this example is that \emph{agents who do not receive large items must be compensated with small items.}
We use the following definition to identify, in general, the agents who may receive no large item.

\begin{definition}[First block]\label{def:firstblock}
\leavevmode Under the top-set competitive rule $x$, let $t_i\deq\min\Bigl\{t:\sum_{s=1}^{t}x_{i,\pi_i(s)}\ge1\Bigr\}$ be the smallest $t$ such that agent $i$'s $t$ most valuable goods have
total probability at least one. We call
$\{\pi_i(1),\ldots,\pi_i(t_i)\}$ her \emph{first block}.
\end{definition}
Taking $t=t_i$ in
\ref{cond:PR}, agent $i$ receives at least
$\bigl\lfloor\sum_{s=1}^{t_i}x_{i,\pi_i(s)}\bigr\rfloor\ge1$ good from her
first block in every outcome. If every good in her first block is worth at
least $\alpha\TPS_i$, this good alone gives her the desired bound, whatever
she receives from the remaining goods. These agents are easy to handle.
What matters is therefore whether each good in the first block is worth at
least $\alpha\TPS_i$, which motivates the following definition.

\begin{definition}[High and low goods]
\label{def:high}
For every agent $i$, let $H_i\deq\{g:v_i(g)\ge\alpha\TPS_i\}$ be her \emph{high good set} and $L_i\deq M\setminus H_i$ be her \emph{low good set}.
\end{definition}

On the other hand, if agent $i$'s high good set $H_i$ is a proper subset of her first block, agent $i$ has a positive probability for not receiving any high good. We call such an agent deficient.
\begin{definition}[Deficient agents]
\label{def:deficient}
Agent $i$ is \emph{deficient} if $|H_i|<t_i$, that is, if $H_i$ is a proper subset of her first block.
Equivalently, agent $i$ is deficient if $\sum_{g\in H_i}x_{ig}<1$.
Let $D$ be the set of all deficient agents.
For each deficient agent $i\in D$, define $p_i\deq1-\sum_{g\in H_i}x_{ig}$ to be the probability of not receiving a high good; for each non-deficient agent $i\in N\setminus D$, define $p_i\deq0$.
\end{definition}
Conditioning on not receiving any high good, we compensate a deficient agent $i$ by letting her receive a relatively large fraction $1/k_i$ of each of her low goods, where the integer $k_i\ge2$ will be specified later in \eqref{eq:parameters}.
In this case, the expected value is $\frac1{k_i}v_i(L_i)$, and Lemma~\ref{lem:prefix} (applied to $L_i$) guarantees that the ex-post value is at least $\frac1{k_i}v_i(L_i)-v_i(g^\ast)$, where $g^\ast$ is her most valuable low good.
The term $v_i(g^\ast)$ is upper bounded by $\alpha\TPS_i$ by the definition of low goods.
If we can ensure sufficiently large $1/k_i$ to make the first term $\frac1{k_i}v_i(L_i)$ large (say, at least $2\alpha\TPS_i$), we have the desired approximation ratio $\alpha$.

We remark that the discussion here is only a high-level illustration of the compensation idea.
However, we cannot ensure that agent $i$ receives the fraction $1/k_i$ of every good in $L_i$: some of these goods are largely taken by other agents, either because many agents nominate them, or because they are reserved for other agents as high goods in the same outcome.
We further allow each deficient agent $i$ to give up a set of items $Q_i$, and show that her compensation value $\frac1{k_i}v_i(L_i\setminus Q_i)$ is still large.

An agent who receives a high good has ex-post approximation ratio at least $\alpha$, since every high good is worth at least $\alpha\TPS_i$.
Whether this can happen in every outcome depends on her total high-good probability $x_i(H_i)\deq\sum_{g\in H_i}x_{ig}$.
In any implementation of $x$, by the union bound,
\begin{equation}\label{eq:union}
 \Pr[A_i\cap H_i\ne\emptyset]\le\sum_{g\in H_i}\Pr[g\in A_i]=x_i(H_i).
\end{equation}
Thus, a deficient agent receives no high good with
probability at least $p_i$ in \emph{every} implementation.

In the next definition, we define the retained high-good suballocation $h$ by truncating the top-set competitive rule $x$.
The suballocation $h$ only assigns positive fractions to agents' high goods.
Each non-deficient agent's $h_i$ sums to exactly $1$, indicating that a high good will be surely allocated in every ex-post allocation.
For each deficient agent $i$, $h_i$ sums to $x_i(H_i)<1$.

\begin{definition}[Retained high-good suballocation]\label{def:h}
For each agent $i$, her \emph{retained high-good suballocation} $h_i\le x_i$ is obtained by the following.
If $i$ is deficient, $h_{ig}=x_{ig}$ for each $g\in H_i$ and $h_{ig}=0$ otherwise. 
If $i$ is not deficient, $h_i$ is obtained by truncating $x_i$ at total probability mass $1$, when items are in descending order of $v_i$. Formally, it is defined as follows, where $t_i$ is as in \cref{def:firstblock}:
\begin{enumerate}
    \item for each $s<t_i$, set $h_{i,\pi_i(s)}=x_{i,\pi_i(s)}$;
    \item for $t_i$, set $h_{i,\pi_i(t_i)}$ such that $\sum_{s=1}^{t_i}h_{i,\pi_i(s)}=1$;
    \item for each remaining item $g$, set $h_{i,g}=0$.
\end{enumerate}
\end{definition}

Every agent $i$ satisfies $h_i(M)=\min\{1,x_i(H_i)\}$, and every good $g$ satisfies $\sum_{i\in N} h_{ig}\le1$ since $h\le x$ and $\sum_{i\in N} x_{ig}=1$.
A non-deficient agent $i$ has $h_i(M)=1$, and she is already satisfied even if no more items are allocated to her.

\begin{example}\label{ex:deficient}
Let $n=4$. By \cref{lem:marginals}, $x_{ig}=(7-f_g)/12$ for every good $g\in T_i$. Suppose
agent $1$'s top set consists of three high goods nominated by one, two, and
four agents, with probabilities $6/12$, $5/12$, and $3/12$. Then
$x_1(H_1)=14/12\ge1$, and $h_1$ keeps $6/12$, $5/12$, and $1/12$. Suppose (in a different profile)
agent $2$'s high goods are her three nominated goods, each nominated by all
four agents. Then $x_2(H_2)=9/12$, so agent $2$ is deficient with $p_2=1/4$.
\end{example}

\leavevmode The first step of our mechanism implements $h$ by a random matching between agents and high goods, in which agent $i$ receives good $g$ with probability $h_{ig}$ and receives at most one high good.
Thus, every non-deficient agent receives a high good with probability one, while a deficient agent $i$ receives none with probability $p_i$.
Later steps then allocate the other goods, given the matching realized here.

We now describe the outcome of the implementation of $h$ by a matching in the bipartite graph.
The left-hand side contains the $n$ agents.
The right-hand side contains $m+|D|$ vertices, where the $m$ vertices represent the $m$ items, and for each deficient agent $i$, a \emph{dummy} item $e_i$ is constructed.
The dummy item $e_i$ serves as a placeholder, and agent $i$ receiving $e_i$ means that it does not receive any high good and needs to be compensated in $L_i$.
An edge $(i,g)$ is in the graph if $h_{ig}>0$, and we build an edge between $i$ and $e_i$ for every deficient agent.
An edge $(i,g)$ is \emph{nominated} if $g\in T_i$, i.e., agent $i$ nominates $g$.

Every outcome of the retained high-good suballocation $h$ is a matching of size $n$.
Other than the matching requirement (each agent receives exactly one item, a physical good or a dummy item; each item is given to at most one agent), we also impose some other requirements for each realized matching.
There are two types of requirements: the first type restricts the set of agents who are matched to their dummy items (constraint \ref{cons:D} in \cref{subsec:reservation}); the second type places restrictions on the set of physical goods that are matched (constraint \ref{cons:C}).
Roughly speaking, both types of requirements ensure that the deficient agents receiving their dummy items can be sufficiently compensated in the later steps of the implementation.
In particular, the first type ensures that there are not too many agents who need compensation, and the second type ensures that, for each deficient agent $i$ matched to her dummy item, the first step gives only a limited number of goods in her low good set $L_i$ to other agents.
The precise requirements will be defined later.
\Cref{sec:implementation} constructs such a matching.

\leavevmode Next, we fix how large a fraction of each low good a deficient agent receives when she is compensated. This fraction has to balance two needs: it must be large enough for her low goods to give her a constant fraction of $\TPS_i$, and small enough that, on every good, the shares of all agents who receive their dummy items add up to at most one. Three parameters govern this balance.
For a deficient agent $i$, define
\begin{equation}\label{eq:parameters}
 a_i\deq n-|H_i|,\qquad q_i\deq np_i,\qquad k_i\deq\max\{2,\lceil2q_i\rceil\}.
\end{equation}
Since a nominated edge has probability at least $1/n$ (\cref{lem:marginals}), $n$ high nominated goods would always reach total probability one; $a_i=n-|H_i|$ counts how many of these $n$ goods are missing, and $q_i=p_i/(1/n)$ is the probability of not receiving a high good, measured in units of $1/n$.
We call $k_i$ the \emph{compensation parameter} of agent $i$ and $1/k_i$ her \emph{compensation fraction}.
When agent $i$ receives her dummy item, she will be compensated with a fraction
$1/k_i$ of every low good outside $Q_i$ that is not reserved for another agent in the first step; the choice $k_i\ge2q_i$ makes
the average share $p_i/k_i$ at most $1/(2n)$, and $k_i\ge2$ makes each share
at most $1/2$. In \cref{ex:deficient}, $p_2=1/4$, $q_2=1$, $a_2=1$, and
$k_2=2$. In the following, Lemmas~\ref{lem:structure}--\ref{lem:nontop} describe these parameters.

\begin{lemma}\label{lem:structure}
For every deficient agent $i$,
\begin{equation}\label{eq:structure}
 H_i\subseteq T_i,\qquad
 q_i=a_i-\frac{\sum_{g\in H_i}(n-f_g)}{n-1},\qquad
 1\le a_i\le n,\quad 0<q_i\le a_i.
\end{equation}
\end{lemma}
\begin{proof}
Suppose some $g\in H_i$ lies outside $T_i$. Every good in $T_i$ is worth at
least $v_i(g)$, so all $n$ goods in $T_i\cup\{g\}$ are high. By
\cref{lem:marginals},
\[
 x_i(T_i\cup\{g\})
 =\frac{(n-1)(2n-1)+(n-1)-\sum_{g'\in T_i\cup\{g\}}f_{g'}}{n(n-1)}
 \ge\frac{2n(n-1)-n(n-1)}{n(n-1)}=1,
\]
because the goods of $T_i\cup\{g\}$ receive at most $n(n-1)$ nominations in
total. This contradicts deficiency. Hence, $H_i\subseteq T_i$, and summing
the nominated-edge formula of \cref{lem:marginals}, written as
$x_{ig}=1/n+(n-f_g)/(n(n-1))$, over $H_i$ gives
\[
 p_i=1-x_i(H_i)=\frac{n-|H_i|}{n}-\frac{\sum_{g\in H_i}(n-f_g)}{n(n-1)},
\]
which is the identity for $q_i=np_i$. Finally, $|H_i|\le n-1$ gives
$a_i\ge1$, deficiency gives $q_i>0$, and $f_g\le n$ gives $q_i\le a_i$.
\end{proof}

\begin{lemma}\label{lem:lowsupply}
For every deficient agent $i$, $v_i(L_i)\ge a_i\TPS_i$.
\end{lemma}
\begin{proof}
Every low good is worth less than $\alpha\TPS_i<\TPS_i$, so capping at
$\TPS_i$ does not change it. By the definition \eqref{eq:tps} of $\TPS_i$,
\[
 v_i(L_i)
 =\sum_{g\in L_i}\min\{v_i(g),\TPS_i\}
 =n\TPS_i-\sum_{g\in H_i}\min\{v_i(g),\TPS_i\}
 \ge(n-|H_i|)\TPS_i=a_i\TPS_i.\qedhere
\]
\end{proof}

\begin{lemma}\label{lem:nontop}
The total retained mass on non-nominated edges is at most one:
\[
 \sum_{i\in N}\sum_{g\in M\setminus T_i}h_{ig}\le1.
\]
\end{lemma}
\begin{proof}
If agent $i$ retains mass on a non-nominated good $g\in H_i$ under $h$, then all her
$n-1$ nominated goods are worth at least $v_i(g)$, hence high, and were
processed first at their full marginals (by the procedure in Definition~\ref{def:h}). By \cref{lem:marginals}, their mass
is at least $(n-1)/n$, so at most $1/n$ non-nominated mass remains for agent $i$. Summing over the $n$ agents gives the claim.
\end{proof}

The three lemmas above explain why deficient agents can be compensated when they receive their dummy items.
Take $\alpha=1/6$ in the discussion below.
First, \Cref{lem:lowsupply} provides a sufficiently large pool of low-good value. 
If all goods in $L_i$ were available, giving agent $i$ a fraction $1/(2a_i)$ of each would give her fractional value at least $\TPS_i/2$.
Since every low good is worth less than $\TPS_i/6$, with the implementation of the remaining goods by \ref{cond:PR}, \Cref{lem:prefix} would then guarantee an ex-post value of at least $\TPS_i/2-\TPS_i/6=\TPS_i/3$, leaving room above the target $\TPS_i/6$ to account for some omitted goods.
Moreover, \Cref{lem:structure} gives $q_i\le a_i$, and hence $k_i=\max\{2,\lceil2q_i\rceil\}\le2a_i$.
Thus, the actual compensation fraction $1/k_i$ is no smaller than the fraction used in this idealized calculation.

The compensation must also respect marginal budgets and simultaneous capacity. 
The choice of $k_i$ ensures two bounds: $p_i/k_i\le1/(2n)$ and $1/k_i\le1/2$.
The first makes compensation affordable on sufficiently unpopular goods, and summed over at most $n$ deficient agents it gives $\sum_{i\in D}p_i/k_i\le1/2$.
The second bounds the share of any single dummy agent by $1/2$.
For capacity, we additionally control which agents receive their dummy items together.
As we will see in the future (\Cref{lem:reservation}), we will place restrictions on the set of deficient agents receiving dummy items such that the following inequality holds:
\[
\sum_{i\in N:\,i\text{ is a dummy agent}}\frac1{k_i}
\le
\max_{i\in D}\frac1{k_i}
+\sum_{i\in D}\frac{p_i}{k_i}
\le1.
\]
Here the first inequality comes from the random matching of the first step (\cref{lem:reservation}), and the second from the choice of the compensation parameters $k_i$; together they make compensation feasible.

In the actual construction, not all low goods are available:
some are reserved for other agents, while others have marginal
budgets too small to support the compensation fraction.
We therefore exclude a suitable set of goods with large numbers of nominations and compensate
agent $i$ on the unreserved low goods outside it.
The nomination identities in \Cref{lem:structure} help bound the
cost of this exclusion, while \Cref{lem:nontop} shows that
non-nominated edges contribute at most one unit of additional
retained mass. Together with the reservation quotas, these bounds
control how many low goods can be unavailable in each outcome.
Since each omitted low good is worth less than $\alpha\TPS_i$,
controlling their number controls the loss of compensation value.

\subsection{Compensation and comparison sets}
\label{sec:compensation}

\leavevmode\Cref{sec:high} described the first step of the mechanism, a random matching that reserves high goods, and fixed the compensation fraction $1/k_i$ of each deficient agent $i$. This subsection designs the second step: it specifies the low goods on which an agent who receives her dummy item is compensated, and the requirements this compensation must meet.
Let $R$ be the set of (physical) goods reserved in an outcome of the matching of \cref{sec:high}.
For every outcome $R$ in the first step, we will define a fractional suballocation on $M\setminus R$, and then use rule \ref{cond:PR} to realize it.
The fractional suballocation and its implementation will guarantee enough compensation and achieve the $\alpha\TPS_i$ guarantee.
Notice that the allocation $M\setminus R$ in this step may still be a suballocation, and we will append it to a complete allocation in the last step.

Below, we first describe the requirements we want for the fractional suballocation of $M\setminus R$ in the second step.
Two kinds of low goods are unsuitable for compensation: goods so popular that $x_{ig}<p_i/k_i$, which would violate (R1), and goods that may be
reserved for other agents in the same outcome. Each deficient agent
therefore gives up an initial segment of a common popularity order.

\begin{definition}[Popularity prefixes]\label{def:prefixes}
Order all goods by decreasing $f_g$, with ties broken by index, and let $P_t$
be the first $t$ goods, for $0\le t\le m$. For $Q\subseteq M$, let
\[
 \rho(Q)=\sum_{i\in N} h_i(M\setminus Q)
\]
be the expected number of reserved goods outside $Q$.
\end{definition}

\begin{definition}[Eligibility and comparison sets]\label{def:comparison}
For a deficient agent $i$, a prefix $P_t$ is \emph{eligible} if
\begin{equation}\label{eq:eligibility}
 x_{ig}\ge p_i/k_i\qquad\text{for all } g\in L_i\setminus P_t.
\end{equation}
Its \emph{cost} for $i$ is
\begin{equation}\label{eq:choose}
 |P_t\setminus H_i|+\lceil\rho(P_t)\rceil .
\end{equation}
Let $\ell_i$ be the smallest $t$ minimizing the cost among eligible prefixes,
and let $Q_i=P_{\ell_i}$ be agent $i$'s \emph{comparison set} and $d_i$ its
cost.
\end{definition}

The full prefix $P_m=M$ is eligible, so $\ell_i$ exists. The cost counts the
low goods that $i$ can lose when she is a dummy agent: those she gives up inside
$Q_i$, and those reserved for others outside $Q_i$, of which there will be
at most $\lceil\rho(Q_i)\rceil$ in every outcome (\cref{lem:reservation}).
Including a popular good in the prefix costs one but protects against its
large reserved mass; leaving an unpopular good outside costs only its small
reserved mass. Different agents may choose different $\ell_i$, but all
comparison sets are prefixes of one common order, hence they are nested and form a laminar family; this is
what allows one matching to cap the reservations outside every $Q_i$
simultaneously.

It is more natural to define $\ell_i$ to be the smallest $t$ minimizing $|P_t\setminus H_i|+\lceil\rho(P_t\cup H_i)\rceil$, i.e., with the second term in \eqref{eq:choose} changed such that $H_i$ is included.
This will be exactly aligned to the cost of the low goods by giving up $P_t$ (in particular, we are counting the cost for \emph{low goods}; the cost for mass retained by the others outside $P_t$ should be $\rho(P_t\cup H_i)$ instead of $\rho(P_t)$).
However, the set of constraints corresponding to $\{P_t\cup H_i\}$ is no longer nested and may not form a laminar family.
This will not guarantee integrality, as the resultant matrix may not be totally unimodular; as a result, there is no guarantee that $h$ can be implemented.
Therefore, we loose our upper bound to $\rho(P_t)$.

When a deficient agent $i$ receives her dummy item,
she receives the fraction $1/k_i$ of each good in $L_i\setminus(Q_i\cup R)$.
The compensation must meet three requirements.
\begin{enumerate}[label=(R\arabic*)]
\item \emph{Marginals.} On every low good $g$, the mean compensation,
at most $p_i/k_i$, must not exceed $x_{ig}$.
\item \emph{Capacity.} In every outcome, the dummy agents' shares of a good must
sum to at most one: $\sum_{i\in N:\,i\text{ is a dummy agent}}1/k_i\le1$.
\item \emph{Value.} In every outcome in which $i$ is a dummy agent, she must end
up with value at least $\alpha\TPS_i$.
\end{enumerate}
(R1) is ensured by the choice of $Q_i$ below; (R2) is ensured by the construction of the matching (\cref{lem:reservation}); and (R3) is the subject of \cref{sec:analysis}.

The mechanism, stated in full in \cref{subsec:mechanism}, has four
steps:
(1) \emph{reserve} high goods by a random matching with exact marginals
$(h,p)$; (2) \emph{compensate} every dummy agent with the fraction $1/k_i$
of each good in $L_i\setminus(Q_i\cup R)$; (3) \emph{round} these fractions to integral
bundles following \ref{cond:PR}; and (4) \emph{complete} the remaining marginals. Steps (1) and (3)
are exact convex decompositions, step (2) does not exceed $x$ on average by
(R1), and step (4) adds the remaining probability.

\section{Implementation via total unimodularity}
\label{sec:implementation}

This section constructs the random matching of
\cref{sec:high} and the rounding and completion steps of
\cref{sec:compensation}. Every step is an exact convex decomposition of a
fractional point, and the structural requirements are imposed as linear
counting constraints that the decomposition preserves (\cref{lem:tu}). 
We first describe the reservation step (precise requirements when implementing $h$), then the rounding step, the completion step, and finally the mechanism.

\subsection{Balanced reservations}
\label{subsec:reservation}

We now write the requirements of the matching of \cref{sec:high} as a polytope.
Use a nonnegative variable $y_{ig}$ for each edge with $h_{ig}>0$; $y_{ig}=1$ means ``$g$ is reserved for $i$'', and a nonnegative variable $z_i$ for each
$i\in D$; $z_i=1$ means ``$i$ receives her dummy item''; set $z_i=0$ for agents outside $D$. Order $D=\{i_1,\ldots,i_{|D|}\}$ by decreasing
$1/k_i$. The \emph{reservation polytope} is defined by
\begin{equation}\label{eq:reservationpoly}
\begin{aligned}
 &\mbox{\customlabel{cons:A}{(A)}}\text{(A)}\quad \sum_{g\in M} y_{ig}+z_i=1 &&(\forall i\in N),\\
 &\mbox{\customlabel{cons:B}{(B)}}\text{(B)}\quad \sum_{i\in N} y_{ig}\le1 &&(\forall g\in M),\\
 &\mbox{\customlabel{cons:C}{(C)}}\text{(C)}\quad \sum_{g\in M\setminus P_t}\sum_{i\in N} y_{ig}\le\lceil\rho(P_t)\rceil
      &&(\forall t\in\{0,1,\ldots,m\}),\\
 &\mbox{\customlabel{cons:D}{(D)}}\text{(D)}\quad \Bigl\lfloor\sum_{j=1}^\ell p_{i_j}\Bigr\rfloor
 \le\sum_{j=1}^\ell z_{i_j}
 \le\Bigl\lceil\sum_{j=1}^\ell p_{i_j}\Bigr\rceil
      &&(\forall \ell\in\{1,\ldots,|D|\}).
\end{aligned}
\end{equation}
Constraints \ref{cons:A} and \ref{cons:B} say that the $0$--$1$ points are exactly the matchings of \cref{sec:high}: every agent
receives one retained high good or her dummy item, and every good is used at most once. 

Constraint \ref{cons:C} caps, for every popularity prefix $P_t$, the number of goods reserved outside $P_t$ by the ceiling of its expectation $\rho(P_t)$; since every comparison set is such a prefix, this caps the reserved goods outside every $Q_i$ simultaneously.
This constraint set is necessary: without it, it is possible that all valued low goods of a dummy agent are reserved by other agents.
More discussion is available in Sect.~\ref{subsec:necessity-popular}.

Constraint \ref{cons:D} keeps the number of dummy agents among the first $\ell$ deficient agents within one of its expectation, so dummy agents cannot cluster.
Notice that this finer control on the set of dummy agents is necessary, while only restricting the number of dummy agents is not enough to obtain a constant factor approximation guarantee of TPS or MMS.
See Sect.~\ref{subsec:necessity-dummy} for relevant discussions.

\begin{lemma}\label{lem:reservation}
There is a finite lottery over $0$--$1$ points of
\eqref{eq:reservationpoly} whose mean is exactly $(h,p)$. In every outcome,
with $R$ the set of reserved goods,
\begin{equation}\label{eq:reservationquota}
 |R\setminus P_t|\le\lceil\rho(P_t)\rceil\qquad(0\le t\le m),
\end{equation}
and
\begin{equation}\label{eq:load}
 \sum_{i\in N:\,i\text{ is a dummy agent}}\frac1{k_i}
 \le\max_{i\in D}\frac1{k_i}+\sum_{i\in D}\frac{p_i}{k_i}
 \le1,
\end{equation}
where an empty maximum is zero.
\end{lemma}
\begin{proof}
\emph{Feasibility.} The point $(h,p)$ satisfies \ref{cons:A} because $h_i(M)+p_i=1$ for every agent; it satisfies \ref{cons:B} because the columns of $h$ have mass at most one; at $(h,p)$, the left side of \ref{cons:C} equals $\rho(P_t)$; and the prefix sums of $p$ lie between their own floors and ceilings, which is \ref{cons:D}.

\emph{Laminarity.} The supports of the \ref{cons:A} rows are pairwise disjoint, forming one laminar family.
For the second family, the \ref{cons:C} supports are the sets of edge-variables outside $P_t$; these are nested as $t$ varies.
Each \ref{cons:B} support consists of the edges at a single good, so it lies either inside or outside each \ref{cons:C} support, and it is disjoint from the others.
The \ref{cons:D} supports form a chain of the variables $z_i$, disjoint from all edge variables.
Hence, \ref{cons:B}, \ref{cons:C}, \ref{cons:D} together form a second laminar family.

The following example with $n=3$ and $m=3$ gives a pictorial view (\cref{fig:laminar}) of the two laminar families formed by \ref{cons:A}, and (\ref{cons:B}, \ref{cons:C}, \ref{cons:D}) respectively.
In this example, all three agents are deficient, the decreasing order of $1/k_i$ gives $1,2,3$, and we assume $0<p_1<1<p_1+p_2<2<p_1+p_2+p_3<3$; items are ordered by decreasing order of numbers of nominations.\footnote{We remark that this is a ``fictitious'' example that serves the purpose of illustrating the two laminar families. Actually, there is no valuation profile that satisfies these descriptions.}
\begin{figure}[t]
\centering
$\begin{array}{crrrrrrrrrrrrrl}
 & & y_{11} &  &  & +y_{12} & & & +y_{13} & & & +z_1 & & &=1\\
 \text{(A)}& &  & y_{21} &  &  &+y_{22} & &  &+y_{23} & &  & +z_2& &=1\\
 & &  &  & y_{31} &  & & +y_{32}&  & & +y_{33}&  & &+z_3 &=1\\
    & &  & &  &  & & & & & &  & & &\\
 \hline
    & &  & &  &  & & & & & &  & & &\\
     & & y_{11} & +y_{21} & +y_{31} &  & & & & & &  & & &\leq1\\
   \text{(B)}  & & & & & y_{12} & + y_{22} & + y_{32} & & & &  & & &\leq1\\
   & & & & & & & & y_{13} &+y_{23} & +y_{33} & & & & \leq1\\
    & &  & &  &  & & & & & &  & & &\\
    & & y_{11} & +y_{21} & +y_{31} & +y_{12} & +y_{22} & +y_{32} & +y_{13} & +y_{23} & +y_{33} &  & & &\leq\lceil\rho(P_0)\rceil\\
    \text{(C)} & &  &  & & y_{12} & +y_{22} & +y_{32} & +y_{13} & +y_{23} & +y_{33} &  & & &\leq\lceil\rho(P_1)\rceil\\
      & &  & &  &  &  &  & y_{13} & +y_{23} & +y_{33} &  & & &\leq\lceil\rho(P_2)\rceil\\
    & &  & &  &  & & & & & &  & & &\\
    &0\leq & & & & & & & & & & z_1 & & &\leq 1\\
   \text{(D)} &1\leq & & & & & & & & & & z_1 & +z_2 & &\leq 2\\
    &2\leq & & & & & & & & & & z_1 & +z_2& +z_3 &\leq 3\\
\end{array}$
\caption[The two laminar families of the reservation polytope for $n=m=3$]{The constraints of the reservation polytope \eqref{eq:reservationpoly} for the example with $n=3$ and $m=3$: the rows of \ref{cons:A} form one laminar family, and the rows of \ref{cons:B}, \ref{cons:C}, and \ref{cons:D} form the other. The horizontal rule separates the two families.}
\label{fig:laminar}
\end{figure}

All variables are at most one by \ref{cons:A}.
By \cref{lem:tu}, the polytope has $0$--$1$ vertices and $(h,p)$ is a convex combination of them; sampling a vertex with its weight gives the lottery. Every vertex satisfies \ref{cons:C}, which is \eqref{eq:reservationquota}.

\emph{Load.} In every outcome, apply the upper bound of
\cref{lem:prefix} with objects $i_1,\ldots,i_{|D|}$, weights $p_{i_j}$,
values $1/k_{i_j}$ (which are nonincreasing), and selected set the dummy
agents; \ref{cons:D} is exactly its prefix hypothesis. This gives the first
inequality of \eqref{eq:load}. Finally, since $k_i\ge2$ and $k_i\ge2q_i$,
\[
 \max_{i\in D}\frac1{k_i}\le\frac12,
 \qquad
 \sum_{i\in D}\frac{p_i}{k_i}
 =\frac1n\sum_{i\in D}\frac{q_i}{k_i}
 \le\frac{|D|}{2n}\le\frac12.\qedhere
\]
\end{proof}

\begin{remark}
The nesting of the comparison sets is essential here: caps \ref{cons:C} for
arbitrary, non-nested sets $Q_i$ would destroy laminarity. Likewise, the
order of \ref{cons:D} must be by decreasing $1/k_i$, so that \cref{lem:prefix}
applies with nonincreasing values. The three orders used in this paper,
the popularity order in \ref{cons:C}, the order of the $k_i$ in \ref{cons:D}, and each
agent's value order in \cref{lem:grid} below, play different roles and
need not agree.
\end{remark}

\subsection{Compensation of dummy agents and its implementation}
\label{subsec:grid}
Given every outcome $R$ in the first step satisfying \Cref{lem:reservation}, we will use the fractional suballocation $b$ such that each agent $i$ previously receiving a dummy item receives a fraction $1/k_i$ of every good in $L_i\setminus (Q_i\cup R)$, and she receives nothing in $(Q_i\cup R)$.
Each agent receiving a high good in $R$ continues to receive it (with fraction $1$), and she receives nothing else.
\Cref{lem:reservation} implies that $b$ is feasible.

The following lemma is the key to the final $\alpha\TPS_i$ ex-post guarantee after $b$ is implemented.
The rule \ref{cond:PR} is used for implementing $b$.

\begin{lemma}\label{lem:grid}
Let $b$ be a fractional suballocation (viewed as a matrix with $n$ rows). Suppose that, for every agent $i$,
every coordinate of row $i$ strictly between zero and one is a multiple of
$1/k_i$ for some positive integer $k_i$. Then there is a finite lottery over
integral suballocations that implements $b$, preserves all
coordinates equal to one, and in which every supported suballocation $A$
satisfies
\begin{equation}\label{eq:gridround}
 \bigl|v_i(A_i)-v_i(b_i)\bigr|
 \le\Bigl(1-\frac1{k_i}\Bigr)
       \max_{g\in M:\,0<b_{ig}<1}v_i(g)
 \qquad\text{for every agent } i,
\end{equation}
where an empty maximum is zero.
\end{lemma}
\begin{proof}
Fix every coordinate equal to one and remove its good; no other row has a
positive coordinate on that good. For each agent $i$, list her remaining
positive coordinates in decreasing order of value as
$g_{i,1},\ldots,g_{i,s_i}$. The \emph{rounding polytope} is given by
nonnegativity on these coordinates, the item capacities
$\sum_{i\in N} y_{ig}\le1$, and
\begin{equation}\label{eq:roundingpoly}
 \Bigl\lfloor\sum_{j=1}^t b_{i,g_{i,j}}\Bigr\rfloor
 \le\sum_{j=1}^t y_{i,g_{i,j}}
 \le\Bigl\lceil\sum_{j=1}^t b_{i,g_{i,j}}\Bigr\rceil
 \qquad(1\le t\le s_i).
\end{equation}
The restricted matrix $b$ is feasible (i.e., inside this polytope). The item supports are pairwise
disjoint; the prefix supports form a chain for each agent and are disjoint
across agents. Item capacities bound all variables by one. Thus,
\cref{lem:tu} gives an exact convex decomposition of $b$ into integral
suballocations satisfying \eqref{eq:roundingpoly}. In each row, apply
\cref{lem:prefix} with weights $b_{i,g_{i,j}}$, values $v_i(g_{i,j})$, and
the goods selected for that agent; its last statement gives
\eqref{eq:gridround}. Restoring the fixed assignments adds the same value
to both sides.
\end{proof}

\subsection{Completing the marginals}
\label{subsec:completion}
After we have implemented $b$, we will allocate the remaining items by a random rule for each ex-post outcome of $b$.
We need to make sure that $b$ can be extended so that the overall marginal probabilities match $\{x_{ig}\}_{i\in N;g\in M}$ while the ex-post guarantee $\alpha\TPS_i$ is still maintained.
The following lemma is regarding this.

\begin{lemma}\label{lem:completion}
If a finite lottery of integral suballocations has mean $\bar b\le x$
coordinatewise and $x$ is complete, it can be extended to a finite lottery
of complete allocations with mean $x$, without removing any earlier
assignment.
\end{lemma}
\begin{proof}
Put $e_{ig}=x_{ig}-\bar b_{ig}\ge0$ and $e_g=\sum_{i\in N} e_{ig}$. Good $g$ is
unassigned with probability $1-\sum_{i\in N}\bar b_{ig}=e_g$. Whenever it is
unassigned and $e_g>0$, assign it to agent $i$ with conditional probability
$e_{ig}/e_g$; its added marginal is exactly $e_{ig}$. If $e_g=0$, the good
is assigned in every outcome of positive probability. Any finite joint
distribution with these conditional marginals completes the allocation, and
no value decreases since all values are nonnegative.
\end{proof}

Rounding is performed \emph{before} completion, so that the reciprocal
grids required by \cref{lem:grid} are still present.

\subsection{The mechanism}
\label{subsec:mechanism}

We now put everything together and state the mechanism for $m\ge n\ge2$. 
Given the reports, it
computes the fractional allocation $x$ of \eqref{eq:x}; for every agent
the TPS $\TPS_i$, the high and low goods $H_i$ and $L_i$, and the retained
high mass $h_i$ (\cref{def:h}); the deficient agents $D$ with their
parameters $p_i$, $q_i$, $a_i$, $k_i$ (\cref{def:deficient} and
\eqref{eq:parameters}); and finally the popularity prefixes, the function
$\rho$, and each comparison set $Q_i$
(\cref{def:prefixes,def:comparison}). All of these are deterministic
functions of the reports. The randomness then enters in four steps.

First, the mechanism \emph{reserves} high goods: it samples a reservation
from the lottery of \cref{lem:reservation}, so that every agent is matched either to one of her retained high goods or to her dummy item, and lets $R$ be the set of reserved goods. Second, it
\emph{compensates} the dummy agents: it forms the fractional suballocation
$b$ that gives every reserved good fully to the agent it is reserved for,
and gives every dummy agent $i$ the fraction $1/k_i$ of each good in
$L_i\setminus(Q_i\cup R)$. By \eqref{eq:load}, these
fractions fit on every unreserved good, and a reserved good has only its
owner, so $b$ is indeed a fractional suballocation; this is requirement
(R2). Third, it \emph{rounds} $b$: the fractional entries of row $i$ lie on
the $1/k_i$ grid, so \cref{lem:grid} applies, and the mechanism samples an
integral suballocation from the lottery it provides. Fourth, it
\emph{completes} the allocation by assigning the remaining goods as in
\cref{lem:completion}, with respect to the mean of the rounded
suballocation.

If $D$ is empty, there is no compensation.
This completes the full description of the mechanism.
We will work out the bound $\alpha\TPS_i$ in the next section.

\section{Proof of the main theorem}
\label{sec:proof}

We prove \cref{thm:main} in two parts, both concerning the final integral
lottery of the mechanism of \cref{subsec:mechanism}. \Cref{sec:marginals} shows that its
marginals are exactly $x$, which gives the ex-ante guarantees, truthfulness-in-expectation
and ex-ante envy-freeness.
\Cref{sec:analysis} shows that the ex-post TPS guarantee is $\alpha=1/6$.

\subsection{Exact marginals and the ex-ante guarantees}
\label{sec:marginals}

We first show that the integral lottery produced by the mechanism
implements $x$; this is requirement (R1) in its global form. View $x_{ig}$ as a
\emph{probability budget} of agent $i$ on good $g$. Each step spends part
of this budget; the first three steps never overspend, and the last step
spends the remainder exactly.

\begin{proposition}\label{prop:marginals}
The lottery of the mechanism implements $x$.
\end{proposition}
\begin{proof}
\emph{Reservation.} By \cref{lem:reservation}, agent $i$ is matched to a
high good $g$ with probability exactly $h_{ig}\le x_{ig}$, since $h$ is
retained from $x$ (\cref{def:h}).

\emph{Compensation.} Now, consider $b$, the fractional suballocation in the second step.
For every pair $(i,g)$ with $g\in L_i$, we have $b_{ig}>0$ only if $g\notin Q_i$.
By our mechanism, $b_{ig}=0$ with probability at least $1-p_i$ (when $i$ receives a physical high good after implementing $h$, or when agent $i$ receives a dummy item and $g$ is reserved for another agent after implementing $h$; the former already has probability $1-p_i$), and $b_{ig}=1/k_i$ with probability at most $p_i$ (when $i$ receives a dummy item \emph{and} $g$ is not reserved by the others after implementing $h$; the former already has probability $p_i$).
Therefore, $g$ is allocated to agent $i$ with probability at most $p_i\cdot (1/k_i)$, which is at most $x_{ig}$ by eligibility \eqref{eq:eligibility}.
Thus, the mean $\bar b$ of $b$ satisfies $\bar b\le x$
coordinatewise.

\emph{Rounding.} \Cref{lem:grid} implements each $b$, so the mean of the rounded suballocation is still $\bar b$.

\emph{Completion.} Since every column of $x$ sums to one, a good is left unassigned with probability equal to the total remaining budget on it, and \cref{lem:completion} spends the remaining budget $x-\bar b$ exactly, without removing any earlier assignment. Hence, the final mean is $x$.
\end{proof}

\begin{corollary}\label{cor:exante}
The mechanism is truthful-in-expectation and ex-ante envy-free.
\end{corollary}
\begin{proof}
By \cref{prop:marginals} the lottery implements $x$, so the claim
follows from \cref{prop:truthful}. In particular, the reported values may be
used freely inside the implementation.
\end{proof}

\subsection{The ex-post guarantee}
\label{sec:analysis}

We now prove requirement (R3) holds for $\alpha=1/6$: every integral allocation in the support gives every agent at least one sixth of her TPS. 

\begin{proposition}\label{prop:value}
Every allocation in the support of the mechanism gives every agent $i$ value at least $\TPS_i/6$.
\end{proposition}

An agent $i$ matched to a high good keeps it, worth at least $\alpha\TPS_i$, so only agents receiving dummy items need attention.
In the remaining part of this section, $i$ is the index of an arbitrary agent who receives a dummy item.

Recall from \cref{def:comparison} that
\begin{equation}\label{eq:di}
 d_i=|Q_i\setminus H_i|+\lceil\rho(Q_i)\rceil.
\end{equation}
In any outcome in which $i$ receives her dummy item, at most $d_i$ of her low goods are absent from $L_i\setminus(Q_i\cup R)$: the low goods inside $Q_i$, and, by \eqref{eq:reservationquota}, at most $\lceil\rho(Q_i)\rceil$ reserved goods outside $Q_i$. 
Each of them is worth less than $\alpha\TPS_i$, so by
\cref{lem:lowsupply}
\[
 v_i\bigl(L_i\setminus(Q_i\cup R)\bigr)\ge\bigl(a_i-\alpha d_i\bigr)\TPS_i .
\]
She receives a fraction $1/k_i$ of this value, and rounding loses less than $(1-1/k_i)\alpha\TPS_i$ by \cref{lem:grid} (since the largest item in $L_i$ is worth at most $\alpha\TPS_i$).
Hence, she ends with at least
\[
 \frac{a_i-\alpha d_i}{k_i}\TPS_i-\frac{k_i-1}{k_i}\alpha\TPS_i
 =\frac{a_i-\alpha d_i-\alpha k_i+\alpha}{k_i}\TPS_i,
\]
which is at least $\alpha\TPS_i$ precisely when
\begin{equation}\label{eq:denominator}
 \alpha (d_i+2k_i-1)\le a_i .
\end{equation}
The rest of this subsection proves that \eqref{eq:denominator} holds for $\alpha=1/6$.
That is,
\begin{equation}\label{eq:denominator6}
 d_i+2k_i-1\le 6a_i .
\end{equation}

Notice that both $d_i$ and $k_i$ in the inequality can depend on $\alpha$.
Specifically, $\alpha$ decides the set of high goods $H_i$ for each agent.
$d_i$ depends on $H_i$; $k_i$ depends on $p_i$, which further depends on $H_i$.

In Sect.~\ref{subsec:tight}, we will show that the approximation guarantee $\alpha=1/6$ is tight in \eqref{eq:denominator}.
In particular, we will present a family of tight examples with $H_i=\emptyset$ for every agent $i$ and every $\alpha\geq1/6$.
In this case, $d_i$ and $k_i$ no longer depend on $\alpha$, and we will show that \eqref{eq:denominator} is tight even in this case.

To prove \eqref{eq:denominator6}, we consider two cases:
\begin{enumerate}
    \item every agent only retains her nominated goods in $h$, i.e., for any agent $j$, $h_{jg}=0$ if $g\notin T_j$;
    \item some agent retains her non-nominated goods in $h$, i.e., there exists an agent $j$ with $h_{jg}>0$ and $g\notin T_j$.
\end{enumerate}

\subsubsection{Case 1: every retained good is nominated}
Let $P=\{g:f_g>\frac{n-1}2\}$ be the set of goods nominated by at least half of the agents.
Clearly, $P$ is one of the prefixes $\{P_t\}_{t=0,1,\ldots,m}$.
We first show that $P$ is eligible (see \eqref{eq:eligibility}) for agent $i$.
Each $g\notin P$ satisfies $f_g\leq\frac{n-1}2$.
By \eqref{eq:xcases}, $x_{ig}\geq\frac{n-1-f_g}{n(n-1)}\geq\frac1{2n}$.
On the other hand, $\frac{p_i}{k_i}=\frac{p_i}{\max\{2,\lceil2q_i\rceil\}}\leq \frac{p_i}{2q_i}=\frac1{2n}$.
Thus, $P$ is eligible.
By our choice of $Q_i$ (see Definition~\ref{def:comparison}), we have
$$d_i=|Q_i\setminus H_i|+\lceil\rho(Q_i)\rceil\leq |P\setminus H_i|+\lceil\rho(P)\rceil.$$
To show \eqref{eq:denominator6}, it then suffices to prove
\begin{equation}\label{eq:denominator6-case1}
    |P\setminus H_i|+\lceil\rho(P)\rceil\leq6a_i-2k_i+1.
\end{equation}

Notice that
$$\rho(P)=\sum_{j\in N}h_j(M\setminus P)=\sum_{j\in N}h_j(L_i\setminus P)+\sum_{j\in N}h_j(H_i\setminus P).$$
Without the ceiling, we have
$$|P\setminus H_i|+\rho(P)=|P\setminus H_i|+\sum_{j\in N}h_j(L_i\setminus P)+\sum_{j\in N}h_j(H_i\setminus P).$$
We will separately bound $|P\setminus H_i|+\sum_{j\in N}h_j(L_i\setminus P)$ and $\sum_{j\in N}h_j(H_i\setminus P)$.

For the first, each $g\in P$ satisfies $f_g\geq\frac{n}2$, which is $1\leq\frac2nf_g$.
Thus, we replace each item in $|P\setminus H_i|$ by $\frac2nf_g$.
Among $\sum_{j\in N}h_j(L_i\setminus P)$, by our case 1 assumption, for each $g\in L_i\setminus P$, $h_{jg}>0$ only if $j$ nominates $g$.
Each nomination reserves at most mass $2/n$ (by \eqref{eq:xcases}), so $\sum_{j\in N}h_j(L_i\setminus P)\leq\frac2n\sum_{g\in L_i\setminus P}f_g$.
Putting together, by noticing $P\setminus H_i=P\cap L_i$ (which implies $(P\setminus H_i)\cup (L_i\setminus P)=L_i$),
\begin{equation}\label{eq:denominator6-case1-a}
    |P\setminus H_i|+\sum_{j\in N}h_j(L_i\setminus P)\leq\sum_{g\in P\setminus H_i}\frac2nf_g+\frac2n\sum_{g\in L_i\setminus P}f_g=\frac2n\sum_{g\in L_i}f_g.
\end{equation}

On the other hand, we will prove that
\begin{equation}\label{eq:denominator6-case1-b}
    \sum_{g\in L_i}f_g=n(a_i-1)+(n-1)(a_i-q_i).
\end{equation}
There are $n-a_i$ high goods for agent $i$.
If all the $n$ agents nominate all these $n-a_i$ goods, each agent has $a_i-1$ nominations left.
In this case, the number of nominations on $L_i$ is $n(a_i-1)$, which is the first term in \eqref{eq:denominator6-case1-b}.
Now, each $g\in H_i$ actually receives $f_g$ nominations, so the total number of nominations missing $H_i$ is $\sum_{g\in H_i}(n-f_g)$, and these nominations will be used to nominates more goods in $L_i$.
Therefore,
$$\sum_{g\in L_i}f_g=n(a_i-1)+\sum_{g\in H_i}(n-f_g).$$
Using the middle equality of Lemma~\ref{lem:structure} proves \eqref{eq:denominator6-case1-b}.

Substituting \eqref{eq:denominator6-case1-b} into \eqref{eq:denominator6-case1-a} and noticing $a_i\geq q_i$ (Lemma~\ref{lem:structure}) give
$$|P\setminus H_i|+\sum_{j\in N}h_j(L_i\setminus P)\leq 2(a_i-1)+\frac2n(n-1)(a_i-q_i)\leq4a_i-2q_i-2.$$
Since $k_i=\max\{2,\lceil2q_i\rceil\}<2+2q_i$, we have
\begin{equation}\label{eq:denominator6-case1-first}
    |P\setminus H_i|+\sum_{j\in N}h_j(L_i\setminus P)< 4a_i-k_i.
\end{equation}

Now, we bound the second term $\sum_{j\in N}h_j(H_i\setminus P)$.
Each $g\notin P$ is not nominated by at least $(n-1)/2$ agents by definition of $P$.
By Lemma~\ref{lem:structure}, the total number of non-nominations on $H_i$ is
$$\sum_{g\in H_i}(n-f_g)=(n-1)(a_i-q_i).$$
Since each item in $H_i\setminus P$ is not nominated by at least $(n-1)/2$ agents and total number of non-nominations in $H_i\setminus P$ is at most $(n-1)(a_i-q_i)$, we have
$$|H_i\setminus P|\leq\frac{(n-1)(a_i-q_i)}{(n-1)/2}=2a_i-2q_i<2a_i-k_i+2.$$
Since $|H_i\setminus P|$ is an integer, this is just $|H_i\setminus P|\leq 2a_i-k_i+1$.
Therefore,
\begin{equation}\label{eq:denominator6-case1-second}
    \sum_{j\in N}h_j(H_i\setminus P)\leq|H_i\setminus P|\leq2a_i-k_i+1.
\end{equation}
Combining \eqref{eq:denominator6-case1-first} and \eqref{eq:denominator6-case1-second}, we have
$$|P\setminus H_i|+\rho(P)<6a_i-2k_i+1,$$
which implies \eqref{eq:denominator6-case1}.

\subsubsection{Case 2: some agent retains a non-nominated good}
In this case, we cannot count $\sum_{j\in N}h_j(L_i\setminus P)$ by the nominations, i.e., the bound $\sum_{j\in N}h_j(L_i\setminus P)\leq\frac2n\sum_{g\in L_i\setminus P}f_g$ no longer holds.
However, in this case, most of the nominations are within $P_{n-1}$, the first $n-1$ items with highest numbers of nominations.

We will first show that the number of nominations outside $P_{n-1}$ is less than $n-1$:
\begin{equation}\label{eq:denominator6-case2-s}
    s\deq\sum_{g\in M\setminus P_{n-1}}f_g=\sum_{g\in P_{n-1}}(n-f_g)<n-1.
\end{equation}
By the case assumption, some agent $j$ has $h_{jg}>0$ for $g\notin T_j$.
Agent $j$ must be non-deficient and, by Definition~\ref{def:h}, we have $x_j(T_j)<1$.
Since
$$x_j(T_j)=\frac{n-1}n+\frac{\sum_{g\in T_j}(n-f_g)}{n(n-1)},$$
we have
$$\sum_{g\in T_j}(n-f_g)<n-1.$$
Since $P_{n-1}$ contains $n-1$ goods with minimum $n-f_g$, we proved \eqref{eq:denominator6-case2-s}.

Next, we will show that
\begin{equation}\label{eq:denominater6-case2-di}
    d_i\leq2a_i+1.
\end{equation}
Given that $k_i=\max\{2,\lceil 2q_i\rceil\}\leq 2a_i$ (by Lemma~\ref{lem:structure}, we have $q_i\leq a_i$, and $a_i\geq1$ is an integer), this proves our objective \eqref{eq:denominator6}.
We further consider two sub-cases: $P_{n-1}$ is eligible for agent $i$, and $P_{n-1}$ is not.

\paragraph{Case 2(i): $P_{n-1}$ is eligible.}
In this case,
$$d_i\leq|P_{n-1}\setminus H_i|+\lceil\rho(P_{n-1})\rceil,$$
and it suffices to show
\begin{equation}\label{eq:denominater6-case2(i)}
    |P_{n-1}\setminus H_i|+\rho(P_{n-1})\leq2a_i+1,
\end{equation}
where the ceiling function can be removed since both $|P_{n-1}\setminus H_i|$ and $2a_i+1$ are integers.

We first show that
\begin{equation}\label{eq:denominater6-case2(i)-first}
    |P_{n-1}\setminus H_i|\leq 2a_i-1.
\end{equation}
Let $r=|H_i\setminus P_{n-1}|$.
Then $H_i\cup P_{n-1}$ contains $n-1+r$ goods.
Each agent nominates exactly $n-1$ goods, so each agent does not nominate at least $r$ goods in $H_i\cup P_{n-1}$.
Therefore, the total number of non-nominations in $H_i\cup P_{n-1}$ is at least $nr$, and we have
\begin{align*}
    nr&\leq \sum_{g\in H_i\cup P_{n-1}}(n-f_g)\\
    &\leq\sum_{g\in P_{n-1}}(n-f_g)+\sum_{g\in H_i}(n-f_g)\\
    &<(n-1)+(n-1)(a_i-q_i)\tag{by \eqref{eq:denominator6-case2-s} and Lemma~\ref{lem:structure}}\\
    &<(n-1)(a_i+1).
\end{align*}
Therefore, $r<a_i+1$, and integrality implies $r\leq a_i$.

On the other hand,
$$|P_{n-1}\setminus H_i|=|P_{n-1}|-(|H_i|-|H_i\setminus P_{n-1}|)=|P_{n-1}|-|H_i|+r=(n-1)-(n-a_i)+r\leq2a_i-1,$$
which proves \eqref{eq:denominater6-case2(i)-first}.

We then show
\begin{equation}\label{eq:denominater6-case2(i)-second}
    \rho(P_{n-1})\leq2.
\end{equation}
To see this, there are $s$ nominations outside $P_{n-1}$, and each nomination carries mass at most $2/n$.
This is a total of $2s/n$.

Since we are in case 2, we also need to count mass given by non-nomination reservations.
For each agent $j$ that has a non-nomination reservation, i.e., $h_{jg}>0$ for some $g\notin T_j$, as shown earlier, we have
$$1>x_j(T_j)=\frac{n-1}n+\frac{\sum_{g\in T_j}(n-f_g)}{n(n-1)}\geq\frac{n-1}n+\frac{\sum_{g\in P_{n-1}}(n-f_g)}{n(n-1)}=\frac{n-1}n+\frac{s}{n(n-1)}.$$
Therefore, the reservation mass for non-nomination of agent $j$ is
$$1-x_j(T_j)\leq1-\left(\frac{n-1}n+\frac{s}{n(n-1)}\right)=\frac{1-s/(n-1)}{n}.$$
Therefore,
\begin{equation}\label{eq:denominator6-case2-temp}
    \rho(P_{n-1})\leq \frac{2s}n+\sum_{j\in N}\frac{1-s/(n-1)}{n}=\frac{2s}n+1-\frac{s}{n-1}=1+\left(1-\frac2n\right)\frac{s}{n-1}<2,
\end{equation}
which proves \eqref{eq:denominater6-case2(i)-second}.

Combining \eqref{eq:denominater6-case2(i)-first} and \eqref{eq:denominater6-case2(i)-second} proves \eqref{eq:denominater6-case2(i)}.

\paragraph{Case 2(ii): $P_{n-1}$ is ineligible.}
If every $g\notin P_{n-1}$ satisfies $f_g\leq (n-1)/2$, then the same argument in Case 1 shows that $P_{n-1}$ is eligible: we have $x_{ig}\geq\frac{n-1-f_g}{n(n-1)}\geq\frac1{2n}$ and $\frac{p_i}{k_i}\leq\frac1{2n}$.
Therefore, for $P_{n-1}$ being ineligible, at least one $g$ satisfies $f_g>(n-1)/2$.
Since the total number of nominations outside $P_{n-1}$ is less than $n-1$ (by \eqref{eq:denominator6-case2-s}), there is exactly one item $g\notin P_{n-1}$ with $f_g>(n-1)/2$, and this is the $n$-th item: $\{g\}=P_n\setminus P_{n-1}$.
As a result, $P_n$ is eligible.
To show \eqref{eq:denominater6-case2-di}, it suffices to show
\begin{equation}\label{eq:denominator6-case2(ii)}
    |P_n\setminus H_i|+\rho(P_n)\leq 2a_i+1.
\end{equation}

Firstly, $P_n$ adds one more good to $P_{n-1}$.
By \eqref{eq:denominater6-case2(i)-first}, we have $|P_n\setminus H_i|\leq 2a_i$.
Secondly, the number of nominations outside $P_n$ is less than
$s-(n-1)/2$, since $s$ is the number of nominations outside $P_{n-1}$ and the $n$-th item receives more than $(n-1)/2$ nominations.
Since $s<n-1$, the number of nominations outside $P_n$ is less than $s/2$.
By counting the total nomination mass and the total non-nomination mass, following the derivation in \eqref{eq:denominator6-case2-temp}, we have
$$\rho(P_n)< \frac2n\cdot\frac{s}2+\sum_{j\in N}\frac{1-s/(n-1)}{n}=\frac{s}n+1-\frac{s}{n-1}<1.$$
Putting $|P_n\setminus H_i|\leq 2a_i$ and the above together proves \eqref{eq:denominator6-case2(ii)}.

\subsection{Putting things together}
\label{subsec:proofmain}

\begin{proof}[Proof of \cref{thm:main}]
The boundary cases are handled in \cref{sec:prelim}. For $m\ge n\ge2$,
the mechanism is truthful-in-expectation and ex-ante envy-free by
\cref{cor:exante}. By \cref{prop:value} and \eqref{eq:tpsmms}, every agent
receives at least $\TPS_i/6\ge\MMS_i/6$ in every outcome.
\Cref{app:implementation} gives a polynomial-time implementation with
support size at most $m(n-1)+1$.
\end{proof}

\subsection{Tightness of $\alpha=1/6$}
\label{subsec:tight}
In the analysis in the previous subsection, in the constraint set \ref{cons:C} in \eqref{eq:reservationpoly}, we have only made use of three particular prefixes:
$$P_{n-1},\quad P_n,\quad\mbox{and}\quad P=\{g:f_g>(n-1)/2\}.$$
Each dummy agent's cost $d_i$ is then relaxed with the comparison set $Q_i$ set to one of the three sets above.
In fact, the guarantee $\alpha=1/6$ still holds if \ref{cons:C} in \eqref{eq:reservationpoly} is changed to only three constraints corresponding to $P_{n-1},P_n$, and $P$ (and Definition~\ref{def:comparison} is changed accordingly so that each $Q_i$ only has three candidates).
It is plausible that the improvement in the approximation guarantee is possible by a finer analysis with all the $m$ prefixes considered (where $H_i$ and the variables $a_i,p_i,q_i,k_i$ depending on it are changed according to some $\alpha>1/6$).
However, we will show that the inequality \eqref{eq:denominator} is tight for $\alpha=1/6$, even if all the prefixes are considered.

In the example below, for each agent $i$ and good $g$, we have $v_{ig}<\frac16\TPS_i$.
Therefore, $H_i=\emptyset$ for any $\alpha\geq1/6$.
As a result,
$$h=0,\quad p_i=1,\quad a_i=n,\quad \mbox{and}\quad k_i=2n.$$
The example uses an even number of agents $n\geq8$, and $m=2n^2+16n-9$ items.
Item set $M$ are partitioned into $M_1,M_2,M_-,M_\varepsilon$, with $|M_1|=|M_2|=n-1$, $|M_-|=2n-1$, and $|M_\varepsilon|=2n^2+12n-6$.
We assume the indices of the items for tie-breaking are defined in the order $M_1,M_2,M_-,M_\varepsilon$.
The valuations are defined in Table~\ref{tab:tightexample}.

\begin{table}[h]
    \centering
\begin{tabular}{lcccc}
\hline
     & $M_1$ & $M_2$ & $M_-$ & $M_\varepsilon$\\
\hline
 Type 1 agents: $1,\ldots,\frac{n}2$    & $\frac{n-1}{6n}$ & $\frac{n-2}{6n}$ & $\frac{n-3}{6n}$  & $\frac1{6n}$ \\
 Type 2 agents: $\frac{n}2+1,\ldots,n$ & $\frac{n-2}{6n}$ & $\frac{n-1}{6n}$ & $\frac{n-3}{6n}$  & $\frac1{6n}$ \\
\hline
\end{tabular}
\caption{The tight example for $\alpha=1/6$.}
    \label{tab:tightexample}
\end{table}

It is straightforward to check that $\MMS_i=\TPS_i=1$ for each agent $i$.
Type 1 agents nominate $M_1$ (i.e., $T_i=M_1$ for any type 1 agent $i$), and type 2 agents nominate $M_2$ (i.e., $T_i=M_2$ for any type 2 agent $i$).

The unique solution to \eqref{eq:reservationpoly} is $y=0$ and $z=1$.
To see this, we have $p_i=1$ for every agent $i$, so \ref{cons:D} implies $z_i=1$ for every agent $i$ (i.e., every agent receives a dummy item).
Then \ref{cons:A} implies $y_{ig}=0$ for each pair $(i,g)$.
Constraints \ref{cons:B} and \ref{cons:C} are satisfied automatically.

We consider type 1 agents.
For each type 1 agent $i$, our fractional rule implies
$$x_{ig}=\left\{
\begin{array}{ll}
    \frac{3n-2}{2n(n-1)} & g\in M_1 \\
    \frac{n-2}{2n(n-1)} & g\in M_2\\
    \frac1n & g\in M_-\cup M_\varepsilon
\end{array}\right..$$
We have $p_i/k_i=1/2n>\frac{n-2}{2n(n-1)}$.
This means that an eligible prefix must contain $M_1\cup M_2$ (each good in $M_2$ fails the eligibility constraint $p_i/k_i\leq x_{ig}$, and $M_1$ precedes $M_2$ in the popularity order since items in $M_1\cup M_2$ receive the same number of nominations and $M_1$ precedes $M_2$ in their indices).
On the other hand, $M_1\cup M_2$ is eligible, since $x_{ig}=\frac1n>\frac{p_i}{k_i}$ for each $g\notin M_1\cup M_2$.
Agent $i$'s comparison set is then $Q_i=P_{2(n-1)}=M_1\cup M_2$ since $t=2(n-1)$ minimizes
$$|P_t\setminus H_i|+\lceil\rho(P_t)\rceil=|P_t|+0=|P_t|$$
among all eligible prefixes (as we argued, these are all $P_t$ with $M_1\cup M_2\subseteq P_t$, or equivalently, all $t$ with $t\geq 2(n-1)$).
Therefore, the cost is given by
$$d_i=2(n-1).$$
Since $k_i=2n$ and $a_i=n$, \eqref{eq:denominator} gives
$$\alpha(6n-3)\leq n.$$
Thus, $\alpha=1/6$ is tight if we consider $n\to\infty$.

To see it more explicitly and to figure out the scope of the tightness for $\alpha=1/6$, we consider the fractional allocation $b$ at the second step of our mechanism.
Since $h=0$ and $R=\emptyset$ for the first step and $Q_i=M_1\cup M_2$, we have
$$b_{ig}=\left\{\begin{array}{ll}
    1/2n & g\in M_-\cup M_\varepsilon \\
    0 & g\in M_1\cup M_2
\end{array}\right..$$
Since $b_{ig}\cdot|M_-|=\frac{2n-1}{2n}<1$, with a positive probability, agent $i$ does not receive any item in $M_-$.
Our decomposition rule guarantees the number of items received by $i$ is at most
$$\left\lceil\frac{|M_-|+|M_\varepsilon|}{2n}\right\rceil=n+7.$$
In the case agent $i$ does not receive any item in $M_-$, her valuation is at most
$$(n+7)\cdot \frac1{6n}=\frac16+\frac7{6n},$$
which tends to $\frac16$ as $n\to\infty$.
Given $\MMS_i=\TPS_i=1$, the ratio $\alpha=1/6$ is tight here.

However, this does not suggest that this is a tight example for \emph{the mechanism}.
In the last step of residual completion, every outcome of $b$ is then completed with the unallocated items among $M_1\cup M_2$ further allocated.
In our analysis, we have ignored the value received by agent $i$ in this phase.
However, in this example, agent $i$'s value increment during the last step (residual complement) may be substantial.
To see this, items in $M_1$ are never allocated before the last step (as similar analysis shows that every type 2 agent's comparison set is $M_1$).
Since $|M_1|=n-1$ and there are $n/2$ type 1 agents, agent $i$ may receive one or two items from $M_1$ by the rule.
If one item from $M_1$ is taken in the last step, agent $1$'s ex-post value approaches $1/3$ for $n\to\infty$.

Therefore, we can only conclude that the ratio $\alpha=1/6$ is tight \emph{for the current mechanism with the current analysis} that ignores the value received by each agent in the last residual completion step. 

\section{Necessity and Motivations of Our Techniques}
\label{sec:necessity}
When implementing the fractional allocation rule $x$ in \eqref{eq:x}, we have justified that deficient agents are those who need special handling: in the case they do not receive a large item, they need to be compensated with large mass on their low goods.
In the first step of our mechanism, when decomposing $h$, we have two restrictions on the outcomes:
\begin{enumerate}
    \item we control the set of deficient agents who receive the dummy items by \ref{cons:D} in \eqref{eq:reservationpoly}, and
    \item we control the items allocated/reserved under $h$ by \ref{cons:C} in \eqref{eq:reservationpoly} so that they are confined to those popular ones (items with high numbers of nominations).
\end{enumerate}

In this section, we will motivate both techniques by showing that they are essential to obtain a constant approximation on TPS and MMS, and we will further explain why they are necessary.
To be more specific, we show that neither family of constraints can simply be dropped: the remaining constraints admit reservation lotteries with outcomes that cannot be completed with a constant guarantee.

\subsection{Necessity of dummy agents control}
\label{subsec:necessity-dummy}
As mentioned before, we need to control the set of dummy agents so that not too many of them need compensation.
At first glance, this seems to be a restriction on the \emph{number} of dummy agents.
On the other hand, the constraint set \ref{cons:D} uses a finer control enforcing that the set of dummy agents cannot cluster.
We will show that this finer control is necessary.
At a high level, the intuition is that we need to make sure not too many agents become dummy agents among those who have low probabilities of becoming dummy agents and small total values on the low goods.

Consider replacing \ref{cons:D} with the following single constraint
\begin{equation}\label{eq:dummy-count-only}
\sum_{i\in D}z_i\le
\left\lceil\sum_{i\in D}p_i\right\rceil,
\end{equation}
which only says that the number of dummy agents cannot exceed its expectation, rounded up.
Consider an instance where there are two types of agents.
The first type of agents have $n-1$ common high goods with value $1$ and many small goods that sum up to value $1$.
These agents have very small $p_i$ (since the top $n-1$ goods are high, $x$ allocates each good to each of them with probability at least $1/n$, and the total mass is already at least $(n-1)/n$, which is quite close to $1$).
Suppose the second type of agents have large $p_i$ so that the expected number of dummy agents, $\sum_{i\in D}p_i$, is large enough.
If all agents of the first type become dummy agents, they are forced to split items from their low goods, which have total value $1$.
This yields a very low guarantee on TPS or MMS.
Therefore, we should properly control the \emph{set} of dummy agents so that not too many of them are from the first type.
This motivates our control constraint set \ref{cons:D}.

An explicit example is given below.
\begin{example}[A total-count constraint does not control compensation demand]
\label{ex:dummy-count-insufficient}
Fix a target approximation factor $\alpha\in(0,1]$.
We construct an instance and an exact reservation lottery satisfying
\ref{cons:A}--\ref{cons:C} of \eqref{eq:reservationpoly} and the weaker constraint \eqref{eq:dummy-count-only}, but having positive-probability outcomes that cannot be extended to
$\alpha$-TPS or $\alpha$-MMS allocations while preserving the reservations and implementing $x$.
The obstruction does not depend on the choice of the compensation parameters $k_i$.

Choose an integer $s>\max\{2,1/\alpha\}$ and set $n=s^2$.
Let $N_1=\{1,\ldots,s\}$ and $N_2=N\setminus N_1$ be the two types of agents.
Partition the goods into three sets $C_0,C_1,F$, with
\[
|C_0|=n-s,\qquad |C_1|=s-1,\qquad |F|=n-s+1.
\]
Thus, $m=2n-s$.
The value of each good in a set is given by
\[
\begin{array}{c|ccc}
 &C_0&C_1&F\\\hline
 i\in N_1&1&1&\dfrac1{n-s+1}\\[2mm]
 i\in N_2&1&\dfrac1s&\dfrac1s
\end{array}
\]
Choose the public good order so that all goods in $C_0$ precede those in $C_1$, which in turn precede those in $F$.

Every agent has $\MMS_i=\TPS_i=1$.
Indeed, a type-$1$ agent uses the $n-1$ singleton goods of
$C_0\cup C_1$ and the bundle $F$ as an $n$-partition with every bundle worth one.
A type-$2$ agent uses the $n-s$ singleton goods of $C_0$ and partitions the $s^2$ goods of $C_1\cup F$ into $s$ bundles of $s$ goods, again each worth one.
Every agent's total value is $n$, so these partitions are optimal for MMS.
Moreover, all singleton values are at most one, so the TPS equation holds at one and cannot hold at any larger value.

All agents nominate the same set $T_i=C_0\cup C_1$.
Consequently, the nomination rule \eqref{eq:x} gives $x_{ig}=1/n$ for every pair $(i,g)$.
Since $1/(n-s+1)<1/s<\alpha\le1$, the high sets defined at the threshold $\alpha\TPS_i$ are
\[
H_i=
\begin{cases}
C_0\cup C_1,&i\in N_1,\\
C_0,&i\in N_2.
\end{cases}
\]
Every agent is deficient, so $D=N$, and the retained suballocation has $h_{ig}=1/n$ on $H_i$ and zero elsewhere.
The probabilities of becoming a dummy agent are therefore
\[
p_i=
\begin{cases}
1/n,&i\in N_1,\\
s/n,&i\in N_2.
\end{cases}
\]
In particular,
\[
\sum_{i\in D}p_i
=\frac{s}{n}+\frac{(n-s)s}{n}
=s-1+\frac1s.
\]
Thus, \eqref{eq:dummy-count-only} permits as many as $s$ dummy agents in an outcome, although the probabilities that the $s$ type-$1$ agents become dummy agents sum to only $1/s$.

Consider the following lottery where there is a positive probability (case 1) where all the $s$ type-1 agents are dummy agents. 
In each case below, choose agents, goods, and injective assignments uniformly among the stated possibilities.
\begin{enumerate}
\item With probability $1/s^2$, match every type-$2$ agent to a different good of $C_0$, and match all type-$1$ agents to their dummy items.
\item With probability $(s-1)/s$, give one uniformly chosen type-$1$ agent one uniformly chosen good of $C_1$.
Give the other $s-1$ type-$1$ agents distinct goods of $C_0$, and assign all remaining goods of $C_0$ to distinct type-$2$ agents. Match all remaining agents to their dummy items.
\item With probability $(s-1)/s^2$, give all $s$ type-$1$ agents distinct goods of $C_0$, and assign all remaining goods of $C_0$ to distinct type-$2$ agents. Match all remaining agents to their dummy items.
\end{enumerate}
These probabilities sum to one, and the assignments are possible because $s\ge3$.

\iffalse
We verify that the lottery implements $h$.
For a fixed type-$1$ agent and a fixed good in $C_1$, the reservation
probability is
\[
\frac{s-1}{s}\cdot\frac1s\cdot\frac1{s-1}
=\frac1{s^2}=\frac1n.
\]
Her probability of receiving a good of $C_0$ is
\[
\frac{s-1}{s}\cdot\frac{s-1}{s}
+\frac{s-1}{s^2}
=\frac{s-1}{s}.
\]
By symmetry among the $s(s-1)$ goods of $C_0$, each such edge
also has reservation probability $1/n$.
Every good of $C_0$ is always reserved, and its total reservation
probability to type-$1$ agents is $s/n$.
By symmetry among the $n-s$ type-$2$ agents, each of their edges
to that good has probability
\[
\frac{1-s/n}{n-s}=\frac1n.
\]
No other real edge is used. Thus, the real-edge marginals are
exactly $h$, and the agent equations give exactly the
probabilities $p_i$ as well.
\fi

Every outcome satisfies \ref{cons:A} and \ref{cons:B}.
The numbers of dummy agents in the three cases are respectively
$s$, $s-1$, and $s$.
Hence, the lottery satisfies even the stronger two-sided total bound
\[
\left\lfloor\sum_{i\in D}p_i\right\rfloor
\le\sum_{i\in D}z_i\le
\left\lceil\sum_{i\in D}p_i\right\rceil.
\]
It also satisfies every constraint in \ref{cons:C}.
In fact, the aggregate retained mass on a good is one on $C_0$, $1/s$ on $C_1$, and zero on $F$, whereas every reserved set $R$ is either $C_0$ or $C_0$ plus one good of $C_1$.
For any set of goods $W\subseteq M$,
\[
\sum_{i\in N} h_i(W)
=|W\cap C_0|+\frac{|W\cap C_1|}{s},
\qquad 0\le\frac{|W\cap C_1|}{s}<1.
\]
It follows that, in every outcome,
\[
\left\lfloor\sum_{i\in N} h_i(W)\right\rfloor
\le |R\cap W|\le
\left\lceil\sum_{i\in N} h_i(W)\right\rceil.
\]
Taking $W=M\setminus P_t$ proves \ref{cons:C} for every popularity prefix.

Finally, we show that every outcome in Case 1 gives MMS guarantee only $1/s=1/\sqrt{n}$.
For each agent $i$ of the first type, we have $v_i(L_i)=1$ and $L_i$ is common for all type-1 agents.
Therefore, at every outcome, at least one agent receives value at most $1/s$, which only achieves $1/s$-MMS since $\MMS_i=1$.

\end{example}

\subsection{Necessity of popular set control}
\label{subsec:necessity-popular}
The popular set control \ref{cons:C} in \eqref{eq:reservationpoly} is also crucial.
Without it, it is possible that all valued low goods of a dummy agent are reserved by the others in the first step of the mechanism.
Specifically, consider an agent who values $n-1$ items at $1$, $w$ items at $1/w$ for some $w<n-1$, and the remaining items at $0$.
The MMS and TPS thresholds for this agent are both $1$.
If this agent receives none of the first $n-1$ items (i.e., she becomes a dummy agent), we should not let the remaining agents take all those $w$ items with value $1/w$ for agent $1$.
Without \ref{cons:C}, this can happen even if we have enforced \ref{cons:D}.

A concrete example demonstrating this phenomenon is given below.

\begin{example}
\label{ex:popular-control}
Fix any target $\alpha\in(0,1]$.
Choose $n\ge5$ with $2/n<\alpha$. Let
\[
C=\{c_0,c_1,\ldots,c_{n-2}\},\qquad
G=\{g_1,\ldots,g_{n-2}\},
\]
and let $Z$ contain $n$ additional goods, so $m=3n-3$.
Agent $1$ values each good in $C$ at $1$, each good in $G$ at $1/(n-2)$, and each good in $Z$ at zero.
For $j=1,\ldots,n-2$, let $j^+$ be its cyclic successor in
$\{1,\ldots,n-2\}$. 
Agent $j+1$ values each good in $(C\setminus\{c_j,c_{j^+}\})\cup\{g_j\}$ at $1$, values $c_{j^+}$ at $2/n$, and values every other good at $1/n$. 
Agent $n$ values $C\setminus\{c_0\}$ at $1$ per good, $c_0$ at $2/n$, and every other good at $1/n$.

Every agent has $\MMS_i=\TPS_i=1$.
Agent $1$ uses the $n-1$ core singletons and the bundle $G\cup Z$.
Each other agent uses her $n-2$ unit singletons and partitions her remaining goods into two bundles worth one: the good worth $2/n$ with $n-2$ goods worth $1/n$, and the remaining $n$ goods.
These partitions establish MMS one because every agent's total value is $n$; the same total and singleton values at most one establish TPS one.
Since $1/(n-2)\le2/n<\alpha\le1$, the high goods are exactly the unit-valued goods.

The nomination sets are
\[
T_1=T_n=C,\qquad
T_{j+1}=(C\setminus\{c_j\})\cup\{g_j\}
\quad(1\le j\le n-2).
\]
Hence, $f_{c_0}=n$, $f_{c_j}=n-1$, $f_{g_j}=1$, and $f_z=0$ for $z\in Z$.
By \eqref{eq:x}, every agent is deficient, every high coordinate is retained in full, and
\[
p_1=\frac1{n(n-1)},\qquad
p_{j+1}=\frac3{n(n-1)},\qquad
p_n=\frac1{n-1},\qquad
\sum_{i\in N} p_i=\frac{4n-5}{n(n-1)}<1.
\]
Thus, every nonempty prefix has bounds zero and one in \ref{cons:D}, regardless of the ordering or the choices of $k_i$.

Below, we show an outcome of $h$ where all goods from $G$ are reserved, and agent $1$ is a dummy agent (whose value is then $0$). 
Reserve $g_j$ for agent $j+1$ for each $j=1,\ldots,n-2$, reserve $c_1$ for agent $n$, and give agent $1$ her dummy item.
This matching, denoted $(y^*,z^*)$, satisfies \ref{cons:A}, \ref{cons:B}, and all of \ref{cons:D} in \eqref{eq:reservationpoly}: every assigned edge belongs to $h$, and there is only one dummy agent.

Moreover, it can occur with positive probability in an exact reservation lottery satisfying these constraints.
Set $\varepsilon=p_1$ and consider
\[
\frac{(h,p)-\varepsilon(y^*,z^*)}{1-\varepsilon}.
\]
This residual point is nonnegative and has every agent-row sum one.
Every real-good column of $h$ has mass at most $1-1/n\le1-\varepsilon$, so the residual real-good columns have mass at most one.
Its total mass on dummy items is $(\sum_{i\in N} p_i-\varepsilon)/(1-\varepsilon)<1$.
Treating all edges to dummy items as edges to a single capacity-one vertex, \cref{lem:tu} decomposes the residual point into matchings with at most one dummy agent.
Together with $(y^*,z^*)$ of probability $\varepsilon$, this gives an exact lottery for $(h,p)$ whose every outcome satisfies \ref{cons:A}, \ref{cons:B}, and \ref{cons:D}.

\iffalse
\paragraph{No compensation can repair this event.}
In the bad event, all goods of $G$ are reserved for other agents.
Agent $1$ cannot subsequently receive a core good in any exact
extension preserving the reservations: for every $g\in C$,
$h_{1g}=x_{1g}$, and hence
\[
\Pr[g\in A_1\text{ but }g\text{ was not reserved for }1]
=x_{1g}-h_{1g}=0.
\]
The only remaining goods she can receive are in $Z$, which she
values at zero. Thus, every positive-probability continuation of
this event gives
\[
v_1(A_1)=0<\alpha\TPS_1=\alpha\MMS_1.
\]
This conclusion holds for arbitrary subsequent compensation and
completion, not just for any particular definition of $k_i$.

\fi

Constraint \ref{cons:C} excludes precisely this bad reservation.
Indeed, $C=P_{n-1}$ and
\[
\rho(C)=\frac{2(n-2)}n<2,
\qquad
|R\setminus C|=n-2>2=\lceil\rho(C)\rceil.
\]
\end{example}

\section{Discussion}
\label{sec:discussion}

We gave a truthful-in-expectation mechanism that, in every realized
allocation, gives each agent at least $1/6$ of her truncated proportional
share, and hence of her maximin share. The mechanism is ex-ante envy-free
and therefore ex-ante proportional. Its fractional rule is ordinal while its
implementation is cardinal, and the use of cardinal information is
necessary, since no ordinal algorithm guarantees more than $1/H_n$ of the
maximin share~\citep{amanatidisBirmpasMarkakis2016truthful}. The factor is
independent of both the number of goods and the number of agents, in
contrast with the deterministic setting, where even for two agents the best
factor achievable by a truthful mechanism is
$1/\lfloor\max\{2,m\}/2\rfloor$~\citep{abcm2017truthful}.

The current analysis does not account for the values obtained in the final residual completion step, and it is tight (Sect.~\ref{subsec:tight}).
Improving the constant with the same mechanism requires sharper analysis.
In addition, with the given top-$(n-1)$ marginal rule \eqref{eq:x}, further improvements may be possible if we slightly change our mechanism by handling the implementation of $x$ more carefully.
For example, we can choose the compensation fraction $1/k_i$ with a more careful choice of $k_i$ that can be personalized (e.g., it can have different expressions for different regimes of $a_i$).
Two questions remain open: the largest TPS factor compatible with the top-$(n-1)$ marginal \eqref{eq:x}, and the largest factor compatible with truthfulness-in-expectation by any mechanism. 
A substantially larger constant may require a different truthful fractional rule or a new way of decomposing its marginals; how close such mechanisms can come to the best MMS guarantees known without incentive constraints is an interesting direction.

\section*{AI Methodology Statement}
\addcontentsline{toc}{section}{AI Methodology Statement}

AI methodology: We used OpenAI's GPT-6 Astra and GPT-5.6 Sol, and
Anthropic's Claude, to assist with this research and with the preparation of
this paper. The OpenAI models helped explore candidate constructions and
draft proofs, including a draft of the mechanism and of the analysis that
establishes the guarantee $\alpha=1/7$. 

The original argument for $\alpha=1/7$ suggested by the AI tools was not in the polytope framework presented in the current version.
It was based on a matching argument instead, and the constraint \ref{cons:C} in \eqref{eq:reservationpoly} uses a single carefully chosen prefix $P_t$.
The authors realized that the initial matching argument can be put into the polytope framework.
With the new framework, the authors realized that including all prefixes $\{P_t\}_{t=1,\ldots,m}$ still preserves total unimodularity.
This enables further improvements beyond $1/7$.
The authors worked out a better bound $2/13$. 
The authors then fed the new ideas to GPT-6 Pro.
GPT-6 Pro found the actual ratio $\alpha=1/6$, and it is tight under the current mechanism with the current analysis.
 
All three tools helped draft and revise parts of the text, search for related work, and
check citations. We checked every argument by hand, simplified and rewrote
the proofs that appear here, edited all AI-assisted text, and verified every
reference against the original source. All authors are human and take full
responsibility for the correctness, originality, and integrity of this
paper, including all content produced with AI assistance. The authors
verified the correctness and originality of all content, including
references.

\bibliographystyle{alpha}
\bibliography{references}

\newpage
\appendix
\section{Finite decomposition and polynomial implementation}
\label{app:implementation}

This appendix gives polynomial-time implementations of both convex
decompositions and of residual completion. We use rational linear
programming as a standard computational primitive, without enumerating MMS
partitions or integral allocations.

\subsection{Computing the TPS}
Sort each agent's singleton values. Between two consecutive distinct
singleton values, the function $\sum_{g\in M}\min\{v_i(g),t\}-nt$ is affine: on an
interval where exactly $s$ goods have value greater than $t$, it equals
\[
 \sum_{g\in M:\,v_i(g)\le t}v_i(g)+(s-n)t.
\]
Inspect the finitely many intervals and endpoints and take the largest
zero; an interval on which the function vanishes identically contributes
its right endpoint. Values above $v_i(M)/n$ cannot be zeros, so the search
is bounded. For rational inputs this uses polynomially many rational
operations and produces a rational share of polynomial encoding length.

\subsection{A direct convex-decomposition procedure}
Both polytopes of \cref{sec:implementation} have polynomially many
variables and inequalities, integral coefficients and right-hand sides,
and $0$--$1$ vertices. The same procedure decomposes a feasible point in
either polytope.

Let $P=\{u:Au\le b\}$ be such a bounded polytope and let $x\in P$; equalities
are expressed as pairs of inequalities. Keep all constraints tight at $x$ as
equalities and choose a vertex $v$ of this face, for example by successive
minimization of the coordinates in a public order. A vertex of a face is a
vertex of $P$, so $v$ is integral. If $x=v$, stop. Otherwise set
\begin{equation}\label{eq:peel}
 \lambda=
 \min_{j:\,b_j-A_jv>0}
       \frac{b_j-A_jx}{b_j-A_jv},
 \qquad
 x'=\frac{x-\lambda v}{1-\lambda},
\end{equation}
where $A_j$ is row $j$ of $A$. Because $v$ lies in the face determined by
the constraints tight at $x$, every ratio in \eqref{eq:peel} is positive.
Boundedness and $x\ne v$ imply that some ratio is less than one: otherwise
the ray $v+t(x-v)$, $t\ge1$, would remain in $P$. Thus, $0<\lambda<1$.

For every constraint,
$b_j-A_jx'=(b_j-A_jx-\lambda(b_j-A_jv))/(1-\lambda)\ge0$, so $x'\in P$. All
constraints tight at $x$ remain tight, and at least one previously slack
constraint becomes tight; since $v$ violates the tightness of that
constraint, the new face is a proper subface. Repeat with $x'$. The face
dimension decreases at every nonterminal step, so the procedure stops after
at most as many steps as there are variables. The identities
$x=\lambda v+(1-\lambda)x'$ give the mixture weights.

Lexicographic vertex selection uses polynomially many linear programs. For
rational $x$, the computations are exact and bit lengths remain polynomial:
if $x$ has common denominator $D$ and row $j$ attains the minimum, put
$L=D(b_j-A_jx)$ and $B=b_j-A_jv$, both integers. The coordinates of $x'$
can be written with common denominator $DB-L$, so one step increases the
bit length of the denominator by at most the encoding length of $B$ plus a
constant. Here $v$ is $0$--$1$ and $A,b$ have polynomial encoding length,
and there are polynomially many steps.

Apply the procedure to $(h,p)$ and to each conditional compensation
matrix, restoring the fixed goods after the latter decomposition. The
resulting lottery has polynomial support and encoding length and retains
the proved value guarantees.

\subsection{A small-support residual completion}
\Cref{lem:completion} does not require independent choices for different
goods. In each integral suballocation, use one fresh uniform number in
$[0,1]$. For every unassigned good $g$, partition $[0,1]$ into agent
intervals of lengths $e_{ig}/e_g$ and assign $g$ according to that same
number. The union of all interval endpoints partitions $[0,1]$ into at most
$m(n-1)+1$ intervals; each interval specifies one deterministic completion,
and its length is its probability. All conditional marginals remain
correct, so this implements $x$ with only a polynomial increase in support.

Finally, complete allocation matrices lie in an affine space of dimension
at most $m(n-1)$, because every good has total allocation one. If the
explicit lottery has more than $m(n-1)+1$ support points, find an affine
dependence among their incidence vectors and move the lottery weights along
this dependence until one becomes zero; the other weights remain
nonnegative, their sum remains one, and the mean matrix remains $x$. Repeat.
These rational linear-algebra steps take polynomial time and never introduce
a new allocation. Thus, fairness is preserved and the final support has size
at most $m(n-1)+1$.

\section{Proof of \texorpdfstring{\cref{lem:tu}}{the two-laminar-families lemma}}
\label{app:tu}

Consider a square submatrix $B$ of the incidence matrix. If two selected
rows of $B$ coincide (for instance, because a set occurs in both families or
twice in one family), then $\det B=0$; so assume the selected rows are
distinct. Restricting the sets to the selected columns preserves
laminarity. Within each family, replace each selected row by itself minus
the rows of its maximal proper selected subsets. In containment order this
is a triangular row transformation with unit diagonal, so it preserves the
determinant, and afterwards the row supports within each family are
disjoint: a column belongs only to the smallest selected set containing it.

Negate the rows coming from the first family. Each column now contains at
most one $+1$ and at most one $-1$. Such a square matrix has determinant
$0$ or $\pm1$: if some column has at most one nonzero, expand along it and
use induction; otherwise every column sums to zero, so the rows are
linearly dependent. Row negation and duplication clearly preserve total
unimodularity, and so does appending unit rows (expand along them).

The constraints of $P$ consist of the incidence rows of
$\mathcal{F}_1\cup\mathcal{F}_2$ (for the upper bounds $u_F$), their
negations (for the lower bounds $\ell_F$), and unit rows and their
negations (for $0\le y_e\le1$), so the constraint matrix of $P$ is totally
unimodular. 
The argument in Footnote 1 shows that every vertex is integral:
at a vertex of $P$, a nonsingular system of active constraints determines the coordinates. 
Its determinant is $\pm1$, so Cramer's rule and the integral right-hand sides give integral coordinates; since
$P\subseteq[0,1]^E$, every vertex is a $0$--$1$ point. The polytope $P$ is bounded, and a nonempty bounded polytope is the convex hull of its vertices; \cref{app:implementation} gives a direct finite decomposition procedure.

\section{Review of Freeman, Witkowski, Vaughan, and Pennock~\citep{freemanWitkowskiVaughanPennock2024equivalence}}
\label{sec:properscoringrule}
In this section, items are divisible, and assume there are at least two agents $n\geq2$.

Freeman, Witkowski, Vaughan, and Pennock~\citep{freemanWitkowskiVaughanPennock2024equivalence} established an equivalence between the truthful fair division mechanisms and the mechanisms in the wangering literature, and the truthful mechanisms in the wangering literature then translate to mechanisms in fair division.
For the purpose of this paper, we will describe the class of mechanisms in the fair division context.

The mechanism first defines a \emph{scoring rule} $s_i$ for each agent $i$ that maps the reported valuation vector $u=(u_{1},\ldots,u_{m})$ to a score vector $s_i(u)=((s_i(u))_1,\ldots,(s_i(u))_m)$ that assigns a real number score in $[0,1]$ to each item.
A scoring rule is \emph{proper} if, for agent $i$, truthfully reporting $u=v_i$ maximizes
$$\langle s_i(u),v_i\rangle=\sum_{g\in M}(s_i(u))_g\cdot v_{ig}.$$
If we view $s_i(u)$ as the allocation vector for agent $i$, i.e., if we set $x_{ig}=(s_i(u))_g$ for $g=1,\ldots,m$, then $\langle s_i(u),v_i\rangle$ becomes the value received by agent $i$ by reporting $u$.
In this case, properness is equivalent to truthfulness.

However, for allocations, we have the item capacity constraints $\sum_{i\in N}x_{ig}=1$, whereas the scores for all agents at each item may not sum up to $1$.
The normalization below is applied.
Below, for reported valuation vectors $(u_1,\ldots,u_n)$, for each agent $i$, we denote $s_{ig}=(s_i(u_i))_g$ (and note that $s_{ig}$ only depends on $u_i$):
\begin{equation}\label{eq:properscoring}
    x_{ig}=\frac1n+\frac1n\left(s_{ig}-\frac1{n-1}\sum_{j\in N\setminus\{i\}}s_{jg}\right).
\end{equation}
We have $x_{ig}\in[0,\frac2n]\subseteq[0,1]$ since $s_{ig}\in[0,1]$ and $\sum_{j\in N\setminus\{i\}}s_{jg}\in[0,n-1]$.
In addition,
$$\sum_{i\in N}x_{ig}=1+\frac1n\left(\sum_{i\in N}s_{ig}-\frac1{n-1}\sum_{i\in N}\sum_{j\in N\setminus\{i\}}s_{jg}\right)=1+\frac1n\left(\sum_{i\in N}s_{ig}-\frac1{n-1}\cdot(n-1)\cdot\sum_{i\in N}s_{ig}\right)=1.$$
The allocation $x=\{x_{ig}\}$ is feasible.

Moreover, each agent $i$'s report $u_i$ can only affect $s_i$, so $x_i$ is an affine transformation of $s_i$.
Reporting $u$ maximizes $\langle s_i(u),v_i\rangle$ if and only if it maximizes $\langle x_i,v_i\rangle$.
Therefore, the mechanism given by \eqref{eq:properscoring} is truthful if every $s_i$ is proper.
This defines a class of truthful fair division mechanisms.

In addition, if the same proper scoring rule $s$ is used for all agents, the mechanism is also envy-free.
To see this, notice that $x_i$ is an affine transformation of $s(v_i)$, and, in particular, for two agents $i$ and $j$,
$$x_{ig}-x_{jg}=\frac{(s(v_i))_g-(s(v_j))_g}{n-1}.$$
Therefore, we have the following inequality which implies envy-freeness:
$$v_i(x_i)-v_i(x_j)=\frac{\langle s(v_i),v_i\rangle-\langle s(v_j),v_i\rangle}{n-1}\geq0,$$
where the last inequality holds due to properness: consider agent $i$ with truthful valuation $v_i$ and report $v_j$ instead.

Our rule \eqref{eq:x} is obtained by \eqref{eq:properscoring} with each $s_i$ defined such that it assigns score $1$ to the top $n-1$ items $T_i$ (ties are broken with item indices) and assigns score $0$ to the remaining items.
It is easy to see that the scoring rule is proper: the number/quota $n-1$ is fixed, and the truthful strategy of each agent $i$ is optimal, as the scoring rule will mark the $n-1$ most valuable items on behalf of agent $i$.

\end{document}